\documentclass[conference]{IEEEtran}

\usepackage[T1]{fontenc}
\usepackage[utf8]{inputenc}

\usepackage{hyperref}
\usepackage{cite}

\usepackage[svgnames]{xcolor}
  \definecolor{diffstart}{named}{Grey}
  \definecolor{diffincl}{named}{Green}
  \definecolor{diffrem}{named}{OrangeRed}

\usepackage{listings}
  \lstdefinelanguage{diff}{
    basicstyle=\ttfamily\small,
    morecomment=[f][\color{diffstart}]{@@},
    morecomment=[f][\color{diffincl}]{+\ },
    morecomment=[f][\color{diffrem}]{-\ },
  }

\usepackage{amsmath,amssymb,amsfonts} %amsmath already included
\usepackage[noend]{algpseudocode}
\usepackage{algorithmicx, algorithm}
\usepackage{graphicx} %graphicx already included
\usepackage{textcomp}

\usepackage{tikz}
\usetikzlibrary{arrows,decorations.markings, positioning, shapes}
\usepackage{tabu}

\usepackage{caption}
\usepackage{subcaption} % for subfigures

\usepackage{mathcommands}

\usepackage{scheduling}

\usepackage[normalem]{ulem}

\usepackage{amsthm}

\newtheorem{thm}{Theorem}%[section]
\newtheorem{prop}[thm]{Proposition}
\newtheorem{lem}[thm]{Lemma}
\newtheorem{cor}[thm]{Corollary}

\theoremstyle{definition}
\newtheorem{defn}[thm]{Definition}

\theoremstyle{remark}

\newtheorem{exmpl}[thm]{Example}
\newcommand{\figurestretchvert}{0.47} % stretching all figures at the same time easily % set yscale=\figurestretchvert in each schedule
\newcommand{\Rtilde}{\widetilde{R}}
\newcommand{\Btilde}{\widetilde{B}}

\newcommand{\Jbb}{\mathbb{J}}

\DeclareMathOperator{\TFP}{TFP}
\DeclareMathOperator{\EL}{EL}

\begin{document}
	
	\title{EDF-Like Scheduling for Self-Suspending Real-Time Tasks}
	
	%% AUTHORS:
	\author{\IEEEauthorblockN{Mario~G\"unzel}
		\IEEEauthorblockA{\textit{TU Dortmund University}\\
			mario.guenzel@tu-dortmund.de}
		\and
		\IEEEauthorblockN{Kuan-Hsun Chen}
		\IEEEauthorblockA{\textit{University of Twente}\\
			k.h.chen@utwente.nl}
		\and
		\IEEEauthorblockN{Jian-Jia~Chen}
		\IEEEauthorblockA{\textit{TU Dortmund University}\\
			jian-jia.chen@cs.tu-dortmund.de}
	}

	\maketitle

	\begin{abstract}
In real-time systems, schedulability tests are utilized to
provide timing guarantees.  However, for self-suspending task
sets, current suspension-aware schedulability tests are limited to Task-Level Fixed-Priority~(TFP) scheduling or Earliest-Deadline-First~(EDF) with
constrained-deadline task systems.  In this work 
we provide a unifying schedulability test for the uniprocessor version of Global EDF-Like (GEL) schedulers and arbitrary-deadline task sets.
A large body of
existing scheduling algorithms can be considered as EDF-Like, such as EDF,
First-In-First-Out~(FIFO), Earliest-Quasi-Deadline-First~(EQDF) and Suspension-Aware EDF~(SAEDF). Therefore, the unifying schedulability test is
applicable to those algorithms. 
Moreover, the schedulability test can be applied to TFP scheduling as well.

Our analysis is the first suspension-aware
schedulability test applicable to arbitrary-deadline sporadic
real-time task systems under Job-Level Fixed-Priority (JFP) scheduling, such as EDF.
Moreover, it is the first unifying suspension-aware
schedulability test framework that covers a wide range of
scheduling algorithms.
Through numerical simulations, we show
that the schedulability test outperforms the state of the art
for EDF under constrained-deadline scenarios. Moreover, we
demonstrate the performance of different configurations under
EQDF and SAEDF. 
	\end{abstract} 
	
	\section{Introduction}
	\label{sec:introduction}

	In real-time systems, jobs (task instances) are released
    recurrently by a real-time task, which has to satisfy its timing
    constraints.  More specifically, each job has to finish no
    later than its absolute deadline, which is its relative deadline
    plus its release time.  To derive a schedule from the jobs
    released by tasks, a \emph{scheduling algorithm} is utilized.  To
    ensure timing correctness, schedulability tests for scheduling
    algorithms have to be provided which guarantee that all deadlines
    are met.  For \emph{self-suspending} task sets, a job may release
    its occupation of the processor before being completed, and may
    wait, for instance due to computation offloading or hardware
    acceleration, until the requested service is completed.  When
    considering self-suspension, providing timing guarantees becomes
    more complex. The main reason is that the classical worst-case
    response time and schedulability analyses, such as the critical
    instant theorem~\cite{liu73scheduling}, Time-Demand Analysis (TDA)~\cite{joseph86responsetimes,lehoczky-1989}, or the demand bound
    function~\cite{baruah_sporadic}, are typically based on an
    assumption that a job, after it is released, is either executed or
    waiting to be executed in the ready queue until it finishes.
    Extending such classical analyses to self-suspending task systems
    is non-trivial, and has been demonstrated to be prone to flaw. In
    the literature, a large number of results analyzing
    self-suspending behavior has recently been reported to be
    flawed, c.f.,~\cite{suspension-review-jj, DBLP:journals/rts/GunzelC20, gunzel2021note}.
	
	In the literature, there are two self-suspension models that are most studied, namely the \emph{segmented} self-suspension model ~\cite{RTCSA-BletsasA05, RTSS-ChenL14, Huang:multiseg,DBLP:conf/rtcsa/PengF16,WC16-suspend-DATE, Kim2016,DBLP:conf/ecrts/ChenHHMB19,ecrts15nelissen} and the \emph{dynamic} self-suspension model~\cite{ECRTS-AudsleyB04, huangpass:dac2015, LiuChen:rtss2014,ChenECRTS2016-suspension,DBLP:conf/ecrts/Devi03,guenzel2020sched_test_edf}.
	In the segmented self-suspension model, the sequence of execution and suspension behavior of all jobs is predefined for each task.
	More specifically, for each segment a suspension or execution upper bound is given and the number of segments is fixed.
	In the dynamic self-suspension model, the segmented structure is
    not predefined. 	The jobs of one task may suspend as often and
    as long as the maximum suspension time is not exceeded. Detailed
    discussions of these two models can be found in the survey paper
    by Chen~et~al.~\cite{suspension-review-jj}. A hybrid self-suspension model was
    proposed by von~der~Br\"uggen~et~al.~\cite{vdBrueggen-RTCSA2017}, which can improve the modelling accuracy of the dynamic
    self-suspension model and increase the flexibility of the segmented self-suspension model.
	In this work, we focus on dynamic self-suspending tasks on a single processor, whereas the scheduling algorithm and (sufficient) schedulability test can be used for segmented self-suspension model as well.
	
	In Task-Level \emph{Fixed-Priority} (TFP) scheduling algorithms, the priority is assigned to tasks, i.e., if one task has a higher priority than another task, then all of its jobs are favored to be executed.
	In \cite{huangpass:dac2015,ChenECRTS2016-suspension,RTCSA-KimCPKH95,ECRTS-AudsleyB04} and \cite[Page 162]{Liu:2000:RS:518501} the problem of finding schedulability tests for dynamic self-suspending tasks under preemptive task-level fixed-priority scheduling has been examined.
	Specifically, in \cite{ChenECRTS2016-suspension} a dominating schedulability test for this scenario has been derived. We note that the classical critical instant theorem does not hold anymore when tasks may suspend and the earlier results in \cite{RTCSA-KimCPKH95,ECRTS-AudsleyB04} have been disproved, c.f.,~\cite{suspension-review-jj}.

	In task-level dynamic-priority scheduling algorithms, the priority of the jobs of one task may differ at different time instants. The study of suspension-aware schedulability tests for task-level dynamic-priority scheduling algorithms has been limited to the Earliest-Deadline-First (EDF) algorithm, in which the priority of a job is specified by its absolute deadline. Devi~\cite{DBLP:conf/ecrts/Devi03} provided a schedulability test for EDF without a proof and has been recently disproved by G\"unzel~and~Chen~\cite{DBLP:journals/rts/GunzelC20}. Liu~and~Anderson~\cite{DBLP:conf/ecrts/LiuA13} and Dong~and~Liu~\cite{DBLP:conf/rtss/DongL16}  studied  global EDF on multiprocessor systems and provided schedulability tests, which are applicable for uniprocessor systems by setting the number of processors to one.  The only dedicated analysis for EDF on uniprocessor systems was provided by G\"unzel~et~al.~\cite{guenzel2020sched_test_edf}. They  provided a schedulability test which significantly improves the previous schedulability tests for uniprocessor systems in \cite{DBLP:conf/ecrts/LiuA13} and \cite{DBLP:conf/rtss/DongL16}. Furthermore, Chen~\cite{RTSS2016-suspension} proved that TFP, EDF, Least-Laxity-First (LLF), and Earliest-Deadline-Zero-Laxity~(EDZL) scheduling algorithms do not have constant speedup factors when the suspension cannot be sped up.
	
	The category of \emph{window-constrained} schedulers, where at each time job priorities are assigned according to a \emph{priority point (PP)}, has been proposed in the literature~\cite{DBLP:conf/rtss/LeontyevA07} originally for multiprocessor scheduling to provide general tardiness bounds.
	Recent results~\cite{DBLP:conf/ecrts/EricksonA12} consider \emph{Global EDF-Like (GEL)} scheduling algorithms, where the priority point of the window-constrained scheduler is induced by the job release and a task specific \emph{relative priority point}.
	The popular task-level dynamic-priority algorithms, such as EDF, First-In-First-Out~(FIFO) and EQDF~\cite{DBLP:conf/rtas/BackCS12} fall into this category.
	In~\cite{DBLP:conf/rtss/LeontyevCA09} a schedulability test for GEL scheduling is provided.
	However, suspension-aware schedulability tests have been
    limited to TFP scheduling and EDF
    scheduling for constrained-deadline sporadic real-time tasks as
    detailed in~\cite{suspension-review-jj}.
    In this work, we provide the first
    unifying suspension-aware schedulability test for uniprocessor EDF-Like (EL) scheduling that
    can be applied to a set of widely used scheduling algorithms and
    arbitrary-deadline task systems. 
  
	% Contribution
	\noindent\textbf{\underline{Contributions}:} % short text
	\begin{itemize}
    \item In Section~\ref{sec:capabilities_and_limitations}, we demonstrate how EDF-Like (EL) scheduling algorithms can
    be configured to behave as EDF, FIFO, EQDF,
    suspension-aware EDF~(SAEDF) and TFP scheduling algorithms.
    \item In Section~\ref{sec:test}, we introduce a unifying schedulability test for uniprocessor EL
      scheduling algorithms, that is applicable to
      arbitrary-deadline task systems. To the best of our knowledge,
      this is the first result that can handle arbitrary-deadline task
      sets under Job-Level Fixed-Priority (JFP) scheduling and cover a wide range of scheduling algorithms in one
      analysis framework for self-suspending task
      systems. 
    \item We present the procedure to implement EL scheduling in
      RTEMS~\cite{rtems} and
      {$\text{LITMUS}^{\text{RT}}$}~\cite{litmusrt} in
      Section~\ref{sec:remark+realization}, followed by numerical
      evaluations in Section~\ref{sec:evaluation}.
      Our evaluation
      results show that our schedulability test outperforms the state
      of the art for EDF and is slightly worse than the schedulability
      test by Chen~et~al.~\cite{ChenECRTS2016-suspension} for
      Deadline-Monotonic (DM) scheduling under constrained-deadline
      scenarios.  Moreover, we demonstrate the performance of
      different configurations under EQDF and SAEDF.
	\end{itemize}

\section{System Model}
\label{sec:system_model}

In this work, we consider a set $\Tbb=\{\tau_1, \dots, \tau_n\}$ of $n$ independent self-suspending sporadic real-time tasks, in a uniprocessor system.
Each task $\tau_i,~i \in \{1,\dots, n\}$ is described by a $4$-tuple $\tau_i = (C_i, S_i, D_i, T_i)$, composed of worst-case execution time (WCET) $C_i \in [0,D_i]$, maximum suspension time $S_i\geq 0$, relative deadline $D_i\geq0$ and minimum inter-arrival time $T_i>0$.
The task $\tau_i$ releases infinitely many jobs, denoted by $\tau_{i,j},~j \in \Nbb$, at time $r_{i,j},~j \in \Nbb$.
Job releases are separated by at least $T_i$ time units, i.e., $r_{i,j+1} \geq r_{i,j} + T_i$.
Each job $\tau_{i,j}$ has to be executed for a certain amount of time $c_{i,j} \in [0,C_i]$ until its absolute deadline $d_{i,j} = r_{i,j} + D_i$.
In addition, each job suspends itself dynamically, i.e., it may suspend itself as often as desired without exceeding the maximum suspension time $S_i$.
We denote by $U_i := \frac{C_i}{T_i}$ the utilization of $\tau_i$ and by $U := \sum_{i=1}^{n} U_i$ the total utilization of $\Tbb$.
A task set $\Tbb$ is a \emph{constrained-deadline} task system if
$D_i \leq T_i$ is ensured for every $\tau_i \in \Tbb$. Otherwise, it is an
\emph{arbitrary-deadline} task system.
We assume a model where time is continuous. However, our results can be applied for discretized time as well.

A scheduling algorithm $\Acal$ specifies the execution behavior of jobs on the processor.
More specifically, they determine at each time, which of the jobs in the ready queue is scheduled by the processor.
In a certain schedule we denote for each job $\tau_{i,j}$ start $s_{i,j}$ and finish $f_{i,j}$ of its execution.
We say that a job $\tau_{i,j}$ is \emph{finished by} time $t$, if $f_{i,j}\leq t$. 
The length of the time interval from release to finish of a job $\tau_{i,j}$ is called the {response time} $R_{i,j} = f_{i,j}-r_{i,j}$.
Of special interest is the \emph{worst-case response time} $R_i = \sup_j R_{i,j}$ of a task $\tau_i$.
A schedule is \emph{feasible}, if all jobs finish before or at their absolute deadline, i.e., $f_{i,j} \leq d_{i,j}$ for all $\tau_{i,j}$ or $R_i \leq D_i$ for all $\tau_i$.
A task set is \emph{schedulable} by a scheduling algorithm $\Acal$ if for each job sequence released by $\Tbb$, $\Acal$ creates a feasible schedule.
In this work, we only consider preemptive, work-conserving scheduling, where job execution may be preempted to execute another job, and the processor executes a job whenever there is one in the ready queue.
In the following, we denote by $\Nbb$, $\Nbb_0$, and $\Rbb$ the sets of natural numbers, non-negative integers, and real numbers.

%\subsection{EDF-Like (EL) Scheduling}

In EDF-Like (EL) scheduling, the priority of each job $\tau_{i,j}$ is based on a job-specific \emph{priority point} (PP) $\pi_{i,j} \in \Rbb$.
More specifically, a job $\tau_{i,j}$ has higher priority than $\tau_{i',j'}$ if  $\pi_{i,j}<\pi_{i',j'}$.
The priority point is induced by the release of the job and a task specific paramater $\Pi_i$ denoted by \emph{relative priority point}, i.e., $\pi_{i,j} = r_{i,j} + \Pi_i$.
As a result, smaller $\Pi_i$ in comparison to the other relative priority points favor the jobs of $\tau_i$ to be scheduled first.
The left hand side of Figure~\ref{fig:exmample_notation} depicts the notation used throughout this work.

\begin{defn}[Priority assignment in EL scheduling]
	\label{defn:priority_rp}
	Let $\tau_{i,j}$ and $\tau_{i',j'}$ be two different jobs obtained by tasks $\tau_i$ and $\tau_{i'}$ in~$\Tbb$.
	Furthermore, let $\Pi_i$ and $\Pi_{i'}$ be the relative priority points of $\tau_i$ and $\tau_{i'}$.
	The job $\tau_{i,j}$ has \emph{higher priority} than $\tau_{i',j'}$ if
	\begin{equation}
		\pi_{i,j} %= r_{i,j} + P_i 
		< \pi_{i',j'} %= r_{i',j'} + P_{i'}
	\end{equation}
	for the priority points $\pi_{i,j} = r_{i,j} + \Pi_i$ and $\pi_{i',j'} = r_{i',j'} + \Pi_{i'}$.
	If the priority points coincide, i.e., $\pi_{i,j} = \pi_{i',j'}$, then the 
	tie is broken arbitrarily.
\end{defn}

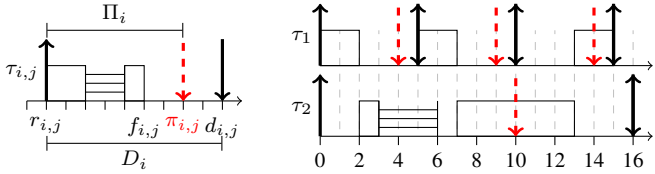
\begin{figure}
	\centering
	\begin{tikzpicture}[yscale=\figurestretchvert, xscale=0.26]
		\footnotesize{}   
		
		\begin{scope}[shift={(-14,1)}]
			\taskname{$\tau_{i,j}$}
			
			\timeline{-1}{10}{}
			\grid{0}{10}{1}{1}
			
			\releases{0}
			\deadlines{9}
			\glprio{7}
			
			\exec{0}{2}
			\exec{4}{5}
			
			\susp{2}{4}
			
			\node[below] at (0,-.25) {$r_{i,j}$};
			\node[below] at (5,-.25) {$f_{i,j}$};
			\node[below] at (7,-.25) {\color{red} $\pi_{i,j}$};
			\node[below] at (9,-.25) {$d_{i,j}$};
			
			\draw[|-|] (0,2) -- node[above] {$\Pi_i$} (7,2);
			\draw[|-|] (0,-1.2) -- node[below] {$D_i$} (9,-1.2);
		\end{scope}
		
		% grid
		\foreach \x in {0,1,...,16}{
			\draw[-,very thin,lightgray, dashed](\x,0.2) -- (\x,3);
		}   
		
		\begin{scope}[shift={(0,2)}] % task one
			\taskname{$\tau_1$}
			
			\timeline{0}{17}{}
			% no labelling
			
			\releases{0,5,10,15}
			\deadlines{5,10,15}
			
			\exec{0}{2}
			\exec{5}{7}
			\exec{13}{15}
			
			\glprio{4,9,14}
			
			% no suspension
		\end{scope}
		
		\begin{scope}[shift={(0,0)}] % task two
			\taskname{$\tau_2$}
			
			\timeline{0}{17}{}
			\labelling{0}{17}{2}{0}
			
			\releases{0,16}
			\deadlines{16}
			
			\exec{2}{3}
			\exec{7}{13}
			
			\susp{3}{6}
			
			\glprio{10}
		\end{scope}
	\end{tikzpicture}
	\caption{Left: Presentation of our Notation. Right: Jobs of two tasks scheduled by EL scheduling. The schedule is feasible and the job priority is given by $\pi_{1,1}<\pi_{1,2} < \pi_{2,1} < \pi_{1,3}$.}\label{fig:exmample_notation}
\end{figure}

Under an assignment of relative priority points $(\Pi_1, \dots, \Pi_n)$, when a job is added to the ready queue, the new highest-priority job is determined and executed.
Whenever a job finishes or suspends itself, it is removed from the ready queue and a new highest-priority job has to be determined and is executed.

\begin{exmpl}
	The right hand side of Figure~\ref{fig:exmample_notation} shows an example schedule obtained by the task set $\Tbb=\{\tau_1,\tau_2\}$, with $\tau_1=(C_1=2, S_1=0, D_1=5, T_1=5)$ and $\tau_2=(C_2=7, S_2=3, D_2=16, T_2=16)$, when using EL scheduling with relative priority points $(\Pi_1=4,\Pi_2=10)$.
\end{exmpl}

We consider that the jobs of a task $\tau_i$ must be executed one
after another in a FIFO manner. If a job $\tau_{i,j}$ of task
$\tau_i$ does not finish execution before the next job
$\tau_{i,j+1}$ of $\tau_i$ is released, job $\tau_{i,j+1}$
cannot be executed even if the processor idles. That is, the
EL scheduling algorithm has to handle this situation as
follows: 1) It sets the priority of the job $\tau_{i,j+1}$ to
$r_{i,j+1} + \Pi_i$ when it arrives at $r_{i,j+1}$, but 2) this
job is only eligible for execution (i.e., ready to be executed) after
all jobs of $\tau_i$ released prior to $\tau_{i,j+1}$ finish execution.

\section{Capabilities and Limitations of EL Scheduling}
\label{sec:capabilities_and_limitations}

EDF-Like (EL) scheduling algorithms are Job-Level Fixed-Priority (JFP) scheduling algorithms, i.e., if one job has higher priority than another job, then it has higher priority at all times.
This is due to the priority definition by comparison of priority points.
Since the priority point of a job does not change once it is assigned, the priorities of jobs are fixed after they are released.
However, EL scheduling covers many of those JFP scheduling algorithms.
These are for example: 
\begin{itemize}
	\item Earliest-Deadline-First (EDF)~\cite{liu73scheduling} ($\Pi_i = D_i$)
	\item Earliest-Quasi-Deadline-First (EQDF)~\cite{DBLP:conf/rtas/BackCS12} \\($\Pi_i = D_i + \lambda C_i$ for some predefined $\lambda \in \Rbb$) 
	\item Suspension-aware variations of EDF (SAEDF)\\($\Pi_i = D_i + \lambda S_i$ for some predefined $\lambda \in \Rbb$) 
	\item First-In-First-Out (FIFO) ($\Pi_i = 0$)
\end{itemize}
As a result, the schedulability test that we present in Section~\ref{sec:test} is applicable to all these scheduling algorithms by configuring the relative priority points $\Pi_i$ accordingly.

Moreover, Task-Level Fixed-Priority (TFP) algorithms can be treated as EL scheduling algorithms as well.
For this purpose, we assume that the tasks are ordered by their priorities, i.e., $\tau_k$ has higher priority than $\tau_{k'}$ if and only if $k<k'$, and $\tau_1$ has the highest priority.
Furthermore, we assume that a worst-case response time upper bound $K_j$ for each task either for the schedule under EL scheduling or under TFP scheduling is given.
If we set the relative priority point of each task $\tau_i$ to $\Pi_i = \sum_{j=1}^{i} K_j$, then EL and TFP coincide.
Please note that the response time 
of tasks may be unbounded and that for such cases a TFP algorithm cannot be treated as EL scheduling algorithm.
However, these cases do not apply in practical scenarios since in real-time systems it is required that the jobs of all tasks finish until their deadline.

\begin{prop}[TFP as EL.]\label{prop:TFPasEL}
	Let $K_1, \dots, K_n \in \Rbb_{\geq 0}$ and $\Pi_i := \sum_{j=1}^{i} K_j$ for $i=1,\dots, n$.
	If for all $j=1,\dots, n$ the value $K_j$ is an upper bound on the worst-case response time of $\tau_j$ in EL or in TFP, then the schedule of $\Tbb$ under EL and the schedule of $\Tbb$ under TFP coincide.
\end{prop}

\begin{proof}
	For an indirect proof we assume that the schedule of $\Tbb$ under EL and TFP does not coincide.
	Let $\tau_{k,\ell}$ be the job with the highest priority in the EL schedule, such that the schedule of $\tau_{k,\ell}$ does not coincide under EL and TFP.
	We define the interval 
	\begin{equation}
		I := [r_{k,\ell}, r_{k,\ell} + K_k).
	\end{equation}
	Let $\Jbb_{\TFP}$ and $\Jbb_{\EL}$ be the set of jobs with higher priority than $\tau_{k,\ell}$ under TFP and under EL that are executed during $I$.
	We will reach a contradiction by showing that $\Jbb_{\TFP}=\Jbb_{\EL}$ and that further the schedule of the jobs in $\Jbb_{\TFP}=\Jbb_{\TFP}$ coincides under TFP and under EL scheduling.
	For that purpose we partition the job sets into three subsets
	$\Jbb_{\TFP} = \coprod_{i \in \setof{+, -, 0}} \Jbb_{\TFP}^i$ and
	$\Jbb_{\EL} = \coprod_{i \in \setof{+, -, 0}} \Jbb_{\EL}^i$,
%	\begin{align}
%		\Jbb_{\TFP} = \coprod_{i \in \setof{+, -, 0}} \Jbb_{\TFP}^i
%		&&\Jbb_{\EL} = \coprod_{i \in \setof{+, -, 0}} \Jbb_{\EL}^i
%	\end{align}
	where each of them denotes the subset of jobs released by tasks of $\Tbb_{+} = \setof{\tau_{k+1}, \dots, \tau_n}$, $\Tbb_{-} = \setof{\tau_1, \dots, \tau_{k-1}}$ or $\Tbb_{0} = \setof{\tau_k}$.
	In the following we show that $\Jbb_{\TFP}^i = \Jbb_{\EL}^i$ for all $i \in \setof{+,-,0}$.
	
	\paragraph{$\Jbb_{\TFP}^+ = \Jbb_{\EL}^+$}
%	We start with the jobs released by tasks $\Tbb_{+} = \setof{\tau_{k+1}, \dots, \tau_n}$.
	Under TFP, there are no jobs of $\Tbb_{+}$ with higher priority than $\tau_{k,\ell}$, i.e., $\Jbb_{\TFP}^+ = \emptyset$.
%	 since all of those jobs have lower priority than $\tau_{k,\ell}$.
	Under EL, we choose any job $\tau_{i,j}$, with $\tau_i \in \Tbb_{+}$, that has higher priority than $\tau_{k,\ell}$, i.e., $\pi_{i,j}=r_{i,j} + \Pi_i \leq \pi_{k,\ell} = r_{k,\ell} + \Pi_k$ holds.
	By subtracting $\Pi_k$ we obtain
	\begin{equation}\label{eq:pf_prop_TFPasEL_1}
	 	r_{i,j} + K_i \leq r_{i,j} + (\Pi_i-\Pi_k) \leq r_{k,\ell}.
	\end{equation}
 	Since $\tau_{i,j}$ has higher priority than $\tau_{k,\ell}$, by assumption the schedule of $\tau_{i,j}$ coincides under EL and TFP, and we have $f_{i,j} \leq r_{i,j} + K_i$.
 	With Equation~\eqref{eq:pf_prop_TFPasEL_1} we conclude $f_{i,j} \leq r_{k,\ell}$.
 	In particular, the job $\tau_{i,j}$ is not executed during $I$.
 	Since $\tau_{i,j}$ was chosen arbitrarily, this means that $\Jbb_{\EL}^+ = \emptyset$ as well.
 	
 	\paragraph{$\Jbb_{\TFP}^- = \Jbb_{\EL}^-$}
 	Under TFP and under EL, jobs of the tasks in $\Tbb_{-}$ %= \setof{\tau_1, \dots, \tau_{k-1}}$ 
 	can only be executed during $I$ if they are released before $r_{k,\ell} + K_k$, 
 	i.e., let $\tilde{\Jbb}^-$ be the set of jobs released before $r_{k,\ell} + K_k$ by tasks of $\Tbb_{-}$ then 
 	$\Jbb_{\TFP}^-, \Jbb_{\EL}^- \subseteq \tilde{\Jbb}^-$.
 	Under TFP, all jobs in $\tilde{\Jbb}^-$ have higher priority than $\tau_{k,\ell}$ since they are released by the tasks of $\Tbb_{-}$.
 	Under EL, we show the same:
% 	need to show that all jobs in $\tilde{\Jbb}^-$ have higher priority than $\tau_{k,\ell}$ as well since in that case by assumption the schedule of these jobs coincides and the same jobs are executed during $I$ under EL and TFP.
 	Let $\tau_{i,j} \in \tilde{\Jbb}^-$, i.e., $r_{i,j} < r_{k,\ell} + K_k$.
 	It directly follows that
 	\begin{equation}
 		\pi_{i,j} = r_{i,j} + \Pi_i 
 		< r_{k,\ell} + K_k + \Pi_i
 		\leq r_{k,\ell} + \Pi_k = \pi_{k,\ell}.
 	\end{equation}
	In particular, $\tau_{i,j}$ has higher priority than $\tau_{k,\ell}$.
	We have shown that all jobs in $\tilde{\Jbb}^-$ have higher priority than $\tau_{k,\ell}$ under EL and under TFP scheduling.
	By assumption the schedule of the jobs in $\tilde{\Jbb}^-$ coincides under TFP and EL.
	Therefore the same jobs of $\tilde{\Jbb}^-$ executed during $I$ under TFP and EL, i.e., $\Jbb_{\TFP}^- = \Jbb_{\EL}^-$.
	
	\paragraph{$\Jbb_{\TFP}^0 = \Jbb_{\EL}^0$}
	Under TFP and under EL, $\tilde{\Jbb}^0 := \setof{\tau_{k,1}, \dots, \tau_{k,\ell-1}}$ are the jobs of $\tau_k$ that have higher priority than $\tau_{k,\ell}$, i.e., $\Jbb_{\TFP}^0, \Jbb_{\EL}^0 \subseteq \tilde{\Jbb}^0$.
	By assumption, the schedule of the jobs in $\tilde{\Jbb}^0$ coincides.
	Therefore the same jobs of $\tilde{\Jbb}^0$ are executed during $I$, i.e., we have $\Jbb_{\TFP}^0 = \Jbb_{\EL}^0$.
	
	We have shown that $\Jbb_{\TFP}=\Jbb_{\EL}$ by the above discussion.
	Since the schedule of the jobs $\Jbb_{\TFP}=\Jbb_{\EL}$ coincides during $I$, $\tau_{k,\ell}$ is preempted/blocked during the same time intervals under TFP and EL scheduling during $I$.
	Hence, the schedule of $\tau_{k,\ell}$ during $I$ coincides.
	Furthermore, since by assumption $K_k$ is an upper bound on the response time under EL or TFP scheduling, the job $\tau_{k,\ell}$ finishes during $I$.
	This proves that the whole schedule of $\tau_{k,\ell}$ coincides.
\end{proof}

Even without knowledge about the worst-case response times, we can use a schedulability test based on EL scheduling for TFP scheduling by setting the relative priority points to $\Pi_i = \sum_{j=1}^i D_j$.
If the schedulability test assures that all jobs meet their deadline, then $D_j$ is an upper bound on the worst-case response time.
In this case, the schedule obtained by EL scheduling coincides with the TFP schedule and is feasible.

\begin{cor}
	If the task set $\Tbb$ is schedulable under EL with $\Pi_i := \sum_{j=1}^{i} D_j, i=1,\dots, n$, then $\Tbb$ is schedulable under TFP as well.
\end{cor}

\begin{proof}
	If $\Tbb$ is schedulable under EL with the given relative priority points $\Pi_i$, then $D_j$ is an upper bound on the worst-case response time of $\tau_j$ for all $\tau_j \in \Tbb$ under EL scheduling.
	In this case, by Proposition~\ref{prop:TFPasEL} the schedule under TFP and EL coincide.
	Therefore, $D_j$ is also an upper bound on the worst-case response time of $\tau_j$ for all $\tau_j \in \Tbb$ under TFP scheduling.
	Hence, $\Tbb$ is schedulable under TFP as well.
\end{proof}

\begin{figure}[t]
	\centering
	\large
	\begin{tikzpicture}[xscale=1, yscale=1]
		\node (a) {TFP};
		\node[right of=a] (b) {$\subseteq$};
		\node[right of=b] (c) {\textbf{EL}};
		\node[right of=c] (d) {$\subseteq$};
		\node[right of=d] (e) {JFP};
		\node[right of=e] (f) {$\subseteq$};
		\node[right of=f] (g) {TDP};
	\end{tikzpicture}
	\caption{Expressiveness of the scheduling policies Task-Level Fixed-Priority (TFP), EDF-Like (EL), Job-Level Fixed-Priority (JFP), and Task-Level Dynamic-Priority (TDP).}\label{fig:inclusions}
\end{figure}

By our discussion in this section, the expressiveness of EDF-Like (EL) scheduling algorithms is between Task-Level Fixed-Priority (TFP) and Job-Level Fixed-Priority (JFP) scheduling.
The inclusions are presented in Figure~\ref{fig:inclusions}.
For the sake of completeness, we include the category of Task-Level Dynamic-Priority (TDP) scheduling algorithms.

The assignment of relative priority points allows to mix different scheduling algorithms, as presented in the following example, or hierarchical scheduling algorithms.
Hence, a schedulability test for EL scheduling is able to handle such cases as well.
\begin{exmpl}
	We consider a task set $\Tbb = \{\tau_1, \dots, \tau_4\}$ of $4$ tasks.
	In the following we demonstrate how to assign priorities such that $\tau_1$ and $\tau_2$ are on one priority-level, and $\tau_3$ and $\tau_4$ are on another priority-level, and on each priority-level EDF is utilized.
	We assign the relative priority points $\Pi_1 = D_1$, $\Pi_2 = D_2$, $\Pi_3 = D_1 + D_2 + D_3$ and $\Pi_4 = D_1 + D_2 + D_4$.
	If $\Tbb$ is schedulable under EL scheduling with the given relative priority points, then EL produces the same schedule as the desired scheduling policy:
	Since $\Pi_i - \Pi_j \geq D_i$ for all $i=3,4$ and $j=1,2$, a job $J$ of $\tau_1$ or $\tau_2$ can only have higher priority than a job $J'$ of $\tau_3$ or $\tau_4$ if $J'$ is already finished by the release time of $J$.
	$\tau_1$ and $\tau_2$ are scheduled according to EDF, since their relative priority points are set to the deadline.
	The tasks $\tau_3$ and $\tau_4$ are also scheduled according to EDF, since the difference between the global priorities $\pi_{3,j}$ and $\pi_{4,j'}$ of each two jobs $\tau_{3,j}$ and $\tau_{4,j'}$ is the same as the difference between the absolute deadlines $r_{3,j} + D_3$ and $r_{4,j'} + D_4$.
  \end{exmpl}

\section{Schedulability Test for EL Scheduling Algorithms}
\label{sec:test}
	
\begin{figure}
	\footnotesize{}
	\centering{}
	\begin{tikzpicture}[yscale=0.6]
		\node[draw, align=left] (exam) at (0,-1) {Examination of Processor States};
		\node[draw, align=left] (backbone) at (0,-2) {Corollary~\ref{cor:backbone}};
		\node[draw, align=left] (bounds) at (3.5,-2.7) {Bounds for:\\
			$B_{k,j}$ $\to$ Lemma~\ref{lem:upper_bound_sum_Bkj} (B1)\\
			$B^i_{k,\ell}$ $\to$ Lemma~\ref{lem:total_bound_interference} (B2)};
		
		\node[draw, align=left] (thm1) at (-1.7,-5) {Fixed analysis window\\Theorem~\ref{thm:suff_sched_test_fixed}};
		\node[draw, align=left] (thm2) at (1.7,-5) {Variable analysis window\\Theorem~\ref{thm:suff_sched_test_var}};
		
		\coordinate (p) at (0,-4);
		
		\draw[->, line width=1pt] (exam) to (backbone);
		\draw[-, line width=1pt] (backbone) to (p);
		\draw[-, line width=1pt] (bounds) -| (p);

		\draw[->, line width=1pt] (p) -| node[above right] {$c\geq r_{k,\ell}$} (thm1);
		\draw[->, line width=1pt] (p) -| node[above left] {$c$ arbitrary} (thm2);
		
	\end{tikzpicture}
	\caption{Roadmap of the proof in Section~\ref{sec:test}.}\label{fig:roadmap}
\end{figure}
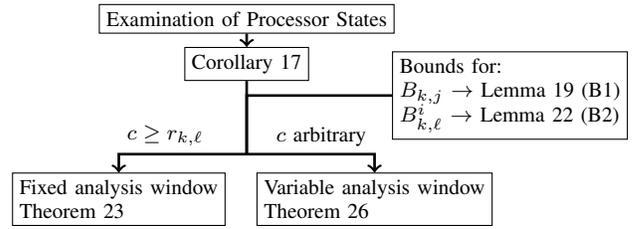
	
In this section, we derive a sufficient schedulability test, i.e., for an arbitrary-deadline task set $\Tbb = \{ \tau_1 , \dots, \tau_n\}$ and an assignment of relative priority points $(\Pi_1, \dots, \Pi_n)$ the test returns \emph{True} if $\Tbb$ is schedulable by the corresponding EL scheduling algorithm.
	
We generate response time bounds by bounding the interference from higher-priority jobs in a certain analysis interval.
In this regard, we start with an examination of different processor states in Section~\ref{sec:processor_states} and derive a response time upper bound in Corollary~\ref{cor:backbone}.
In Section~\ref{sec:estimation_B}, we transfer the response time bound to the case with EL scheduling by approximating two main terms $B_{k,j}$ and $B^i_{k,\ell}$.
We derive two schedulability tests in Theorem~\ref{thm:suff_sched_test_fixed} and Theorem~\ref{thm:suff_sched_test_var} considering two different approaches to define the analysis window.
Algorithm~\ref{alg:sched_test_fixed} and Algorithm~\ref{alg:sched_test_var} depict an implementation of the schedulability test using a simple search algorithm.
The roadmap of our analysis is depicted in Figure~\ref{fig:roadmap}.

\begin{defn}[Active and Current Job]
For a certain schedule, a job $\tau_{k,\ell}$ of a task $\tau_k$ is \emph{active} at time $t$,
%\uline{\textbf{active}} 
if it is released but not already finished by time $t$, i.e., $t \in [r_{k,\ell}, f_{k,\ell})$.
When there are active jobs of task $\tau_k$ at a time instant $t$, then we call the active job of $\tau_k$ with the earliest release the
\emph{current}
%\uline{\textbf{current}} 
job of $\tau_k$ at time $t$.
We call the task $\tau_k$ active at $t$, if there exists an active job of $\tau_k$ at $t$.
\end{defn}

Moreover, the following definitions are used to describe the states of the processor for our analysis.

\begin{defn}[Work, Suspend and Wait]
	~
	\begin{itemize}
		\item  The processor is \emph{working on} a job $\tau_{i,j}$ at time $t$, if $\tau_{i,j}$ is executed on the processor at $t$.
		It is \emph{suspended by} $\tau_{i,j}$ at time $t$, if $\tau_{i,j}$ is suspending itself at $t$, i.e., the remaining suspension time of $\tau_{i,j}$ is reduced.
		\item We say that the processor is \emph{working} at time $t$ if it is working on any job at $t$.
		It is \emph{suspended} at $t$, if it is suspended for at least one job but not working on any job at $t$.
		It is \emph{waiting} at $t$, if it is neither working nor suspended at $t$. The processor is \emph{idle} at $t$, if it is not working at $t$, i.e., if it is suspended or waiting.
	\end{itemize}
\end{defn}

For unambiguous partition of the processor to the different states, we use half-opened intervals, e.g., if the processor is working on a job $\tau_{i,j}$ from time $t_1$ to time $t_2$, then we say that \emph{the processor is working on $\tau_{i,j}$ during $[t_1,t_2)$}.

\subsection{Examination of Processor States}
\label{sec:processor_states}
	
	\begin{figure}
		\centering
		\begin{tikzpicture}[yscale=0.7]
			\footnotesize{}
			\node at (0,1) (root) {};
			
			\node[draw, rectangle, align=center] at (0,0) (a) {Is there an active\\job of $\tau_k$?};
			
			\node[draw, rectangle, align=center, anchor=center] at (2,-1.5) (b) {Is the processor working\\on a job with higher priority\\than the current job of $\tau_k$?};
			
			\node[draw, ellipse, align=center, anchor=center] at (-2,-1.) (ps1) {\textbf{PS~1}};
			
			\node[draw, ellipse, align=center, anchor=center] at (-.5,-2.5) (ps2) {\textbf{PS~2}};
			
			\node[draw, ellipse, align=center, anchor=center] at (4.5,-2.5) (ps3) {\textbf{PS~3}};
			
			\draw[->] (root) -- (a);
			
			\draw[->] (a) -| node[above left] {yes} (b);
			
			\draw[->] (a) -| node[above right] {no} (ps1);
			
			\draw[->] (b) -| node[above right] {no} (ps2);
			
			\draw[->] (b) -| node[above left] {yes} (ps3);
			
		\end{tikzpicture}
		\caption{Decision tree to determine the current processor state (PS) from the perspective of $\tau_k$.}\label{fig:dec_tree}
	\end{figure}
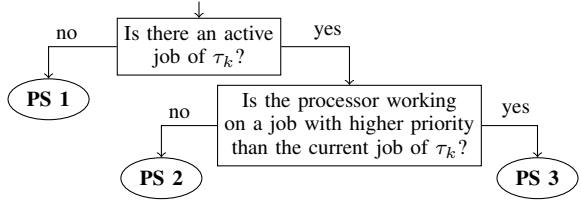

	We consider a schedule obtained by the EL scheduling algorithm with relative priority points $(\Pi_1, \dots, \Pi_n)$ for the task set $\Tbb = \{ \tau_1 , \dots, \tau_n\}$.
	Let $\tau_k \in \Tbb$ be a task. 
	Based on the terminology for different processor states outlined in Section~\ref{sec:system_model}, we distinguish all \textbf{Processor States} (PS) as depicted in Figure~\ref{fig:dec_tree}.
	\begin{itemize}
		\item \textbf{PS~1}: There is no active job of $\tau_k$.
		\item \textbf{PS~2}: There is an active job of $\tau_k$ and the processor is \emph{not} working on a job with higher priority than the current job of task $\tau_k$. 
		\item \textbf{PS~3}: There is an active job of $\tau_k$ and the processor is working on a job with higher priority than the current job of task $\tau_k$.
	\end{itemize}
	In particular, if the processor is in state PS~2, then it is either working on or suspended by the current job of $\tau_k$ since the underlying scheduling algorithm is work-conserving.
	
	\begin{lem}\label{lem:processor_states}
		At each time, the processor is in exactly one of the above three states.
	\end{lem}
	\begin{proof}
		This directly follows from the correctness of the decision tree in Figure~\ref{fig:dec_tree}.
		At each time, the decision tree determines exactly one processor state.
	\end{proof}

	We describe the amount of time that the processor spends in the above states using the following notation.
	
		\begin{table}
		\centering{}
		\begin{tabular}{|l|l|}
			\hline
			Term & Amount of time where processor is ... \\
			\hline
			& \\[\dimexpr-\normalbaselineskip+1pt]
			$\Btilde_k$ & ... in PS~1.\\
			\hline
			& \\[\dimexpr-\normalbaselineskip+1pt]
			$B_k$ & ... in PS~2.\\
			$B_{k,\ell}$ & 
			\begin{tabular}{@{}l@{}}
				... working on/suspended by $\tau_{k,\ell}$ \\
				and not working on higher priority jobs.\\
			\end{tabular}
			\\
			\hline
			& \\[\dimexpr-\normalbaselineskip+1pt]
			$B^i_k$ & ... in PS~3 working on a job of $\tau_i$.\\
			$B^i_{k,\ell}$ & ... working on a job of $\tau_i$ with higher priorty than $\tau_{k, \ell}$.\\
			\hline
		\end{tabular}
		\caption{Notions for time in different states as described in Definition~\ref{def:B_interferences} and Definition~\ref{def:B_interferences_refinement}.}\label{table:summary_notation}
	\end{table}
	
	\begin{defn}\label{def:B_interferences}
		Let $[c,d)$ be any half opened interval with $c<d$.
		\begin{itemize}
			\item $\Btilde_k(c,d)$ denotes the amount of time during $[c,d)$ where the processor is in state \textbf{PS~1}.
			\item $B_k(c,d)$ is the amount of time during $[c,d)$ in \textbf{PS~2}.
			\item For $i \neq k$, $B_k^i(c,d)$ is the amount of time during $[c,d)$ where the processor is in \textbf{PS~3} working on $\tau_i$.
			More specifically, $\sum_{i\neq k}B_k^i(c,d)$ is the total time during $[c,d)$ in \textbf{PS~3}.
		\end{itemize}
		If $c\geq d$, we set all of them, $\Btilde_k(c,d), B_k(c,d), B_k^i(c,d)$, to $0$.
	\end{defn}
	
	The definitions are collected in Table~\ref{table:summary_notation}.
	In the following Lemma, the result from Lemma~\ref{lem:processor_states} is formalized using the previous notation.
	The time in PS~1, PS~2 and PS~3 during some interval $[c,d)$ adds up to the whole interval length.

	\begin{lem}\label{lem:distribution_ps_with_b}
		For any half opened interval $[c,d)$ with $c<d$ we have
		\begin{equation}
			\Btilde_k(c,d) + B_k(c,d) + \sum_{i \neq k} B^i_k(c,d) =(d-c).
		\end{equation}
	\end{lem}
	
	\begin{proof}
		This follows from Lemma~\ref{lem:processor_states}.
	\end{proof}

	For analysis of the worst-case response time, the terms $B_k(c,d)$ and $B^i_k(c,d)$ are inconvenient because at each time, the \emph{current} job of $\tau_k$ has to be determined to calculate those values.
	In the following, we prove that they can be expressed by the more convenient terms $B_{k,\ell}^i(c,d)$ and $B_{k,\ell}(c,d)$ which take a specific job $\tau_{k,\ell}$ as current job into consideration.
	
	\begin{defn}\label{def:B_interferences_refinement}
		Let $[c,d)$ with $c<d$ be any half opened interval.
		\begin{itemize}
			\item $B_{k,\ell}^i(c,d)$ is the amount of time during $[c,d)$ that the processor is working on jobs of $\tau_i$ with higher priority than $\tau_{k,\ell}$.
			\item $B_{k,\ell}(c,d)$ is the amount of time during $[c,d)$ that the processor is working on or suspended by $\tau_{k,\ell}$ while it is not working on a higher-priority job.
		\end{itemize}
		If $c\geq d$, we set all of them to $0$ for simplicity.
	\end{defn}

	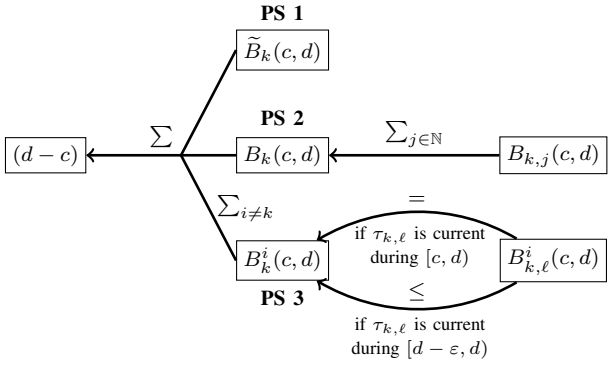
\begin{figure}
		\begin{tikzpicture}[xscale=0.9, yscale=0.7]
			\footnotesize
			\node[draw] (start) at (0.5,0) {$(d-c)$};

			\node[draw] (1) at (4,2) {$\Btilde_k(c,d)$}; 
			\draw (1) node[yshift=0.5cm] {\textbf{PS 1}}; 
			\node[draw] (2) at (4,0) {$B_k(c,d)$};
			\draw (2) node[yshift=0.5cm] {\textbf{PS 2}}; 
			\node[draw] (3) at (4,-2) {$B^i_k(c,d)$};
			\draw (3) node[yshift=-0.5cm] {\textbf{PS 3}}; 
			
			\node[draw] (21) at (8,0) {$B_{k,j}(c,d)$};
			\node[draw] (31) at (8,-2) {$B^i_{k,\ell}(c,d)$};
			
			\draw[<-, line width=1pt] (start) -- node[above] {$\sum$} (2);
			\draw[-, line width=1pt] (2.5,0) -- (1.west);
			\draw[-, line width=1pt] (2.5,0) -- node[right] {$\sum_{i\neq k}$} (3.west);
			
			\draw[<-, line width=1pt] (2) -- node[above] {$\sum_{j\in \Nbb}$} (21);
			
			\draw[<-, line width=1pt] (3) to [out=35, in=145] node[above] {$=$} node[below = 2.5pt, align=center]{\scriptsize{} if $\tau_{k,\ell}$ is current\\ \scriptsize{} during $[c,d)$} (31);
			\draw[<-, line width=1pt] (3) to [out=-35, in=-145] node[above] {$\leq$} node[below, align=center]{\scriptsize{} if $\tau_{k,\ell}$ is current\\ \scriptsize{} during $[d-\eps,d)$} (31);
		\end{tikzpicture}
		\caption{Relations proven in Lemma~\ref{lem:distribution_ps_with_b},~\ref{lem:refinement_Bk} and~\ref{lem:refinement_Bik}.}\label{fig:relations_from_lemmas}
	\end{figure}

	Again, the definition is summarized in Table~\ref{table:summary_notation}.
	In Lemma~\ref{lem:refinement_Bk} and~\ref{lem:refinement_Bik} we will prove the relations between the terms from Definition~\ref{def:B_interferences_refinement} and Definition~\ref{def:B_interferences} as summarized in Figure~\ref{fig:relations_from_lemmas}.
	
	Intuitively, the following lemma states that we can compute the total time in PS~2, by adding up the time in PS~2 for all possible current jobs $\tau_{k,j}$.
	 
	\begin{lem}\label{lem:refinement_Bk}
		For any half opened interval $[c,d)$ we have
		\begin{equation}
			B_k(c,d) = \sum_{j \in \Nbb} B_{k,j}(c,d).
		\end{equation}
	\end{lem}
	\begin{proof}
		The lemma follows from the equivalence between: 1)~the processor is in PS~2 and 
		2)~the processor is \emph{working on or suspended by some job of $\tau_k$ while it is not working on higher-priority jobs}.
		
		Since the processor is work-conserving, we obtain an equivalence between 1) and \emph{the processor is working on or suspended by the current job of $\tau_k$ while it is not working on a higher-priority job}.
		By the definition of the current job, if the processor is working on or suspended by a job, then this is the current job by default. 
		This concludes the equivalence between 1) and 2).
		
		By definition, the time of 1) during $[c,d)$ is $B_k(c,d)$, and the time of 2) during $[c,d)$ is $\sum_{j \in \Nbb} B_{k,j}(c,d)$.
	\end{proof}
	
	In the subsequent lemma, we approximate the time in PS~3 by determining the current job of $\tau_k$ at the end of the analysis window $[c,d)$.
	More specifically, we use that if a job of $\tau_i$ has higher priority than the current job of $\tau_k$ in the analysis window, then it has higher priority than the last current job in the analysis window as well.
	
	\begin{lem}\label{lem:refinement_Bik}
		Consider any half opened interval $[c,d)$.
		If $\tau_{k,\ell}$ is the current job of $\tau_k$ during the whole interval $[c,d)$, then
		\begin{equation}
			B^i_k(c,d) = B^i_{k,\ell}(c,d).
		\end{equation}
		If $\tau_{k,\ell}$ is the current job of $\tau_k$ during $[d-\eps, d)$ for some $\eps >0$, then
		\begin{equation}
			B^i_k(c,d) \leq B^i_{k,\ell}(c,d).
		\end{equation}
	\end{lem}
	\begin{proof}
		By definition, $B^i_k(c,d)$ is the amount of time during $[c,d)$ that the processor is working on a job of $\tau_i$ with higher priority than the current job of $\tau_k$.
		If $\tau_{k,\ell}$ is the current job of $\tau_k$ during the whole interval, then $B^i_k(c,d) = B^i_{k,\ell}(c,d)$.
		If $\tau_{k,\ell}$ is current during $[d-\eps,d)$, then any job of $\tau_k$ which is current at some time during $[c,d)$ has a priority higher than or equal to the priority of $\tau_{k,\ell}$.
		If the processor is working on a job of $\tau_i$ with higher priority than the current job of $\tau_k$ at that time, then that job of $\tau_i$ has also higher priority than $\tau_{k,\ell}$.
		As a result, $B^i_k(c,d) \leq B^i_{k,\ell}(c,d)$. 
	\end{proof}
	
	By definition, $B_k^i$, $B_k$, $\widetilde{B}_k$, $B_{k,\ell}^i$ and $B_{k,\ell}$ are \emph{non-negative} and \emph{additive}, in the sense that $B(c,d)\geq 0$ and $B(c, d) + B(d,e) = B(c,e)$ for all $B \in \{B_k^i, B_{k}, \widetilde{B}_k, B_{k,\ell}^i, B_{k, \ell}\}$ and $c\leq d \leq e$.
	The following two lemmas build a bridge between the processor state formulation and the ability of a job $\tau_{k,\ell}$ to be finished by the end of the analysis window $[c,d)$.
	
	\begin{lem}\label{lem:Bkl_finished}
%		Consider an interval $[c,d)$.
		If $B_{k,\ell}(c,d) \geq C_k + S_k$, then the job $\tau_{k,\ell}$ is finished by time $d$.
	\end{lem}
	\begin{proof}
		Since $B_{k,\ell}(c,d) \geq C_k + S_k$, the processor is working on the job $\tau_{k,\ell}$ for at least $C_k$ time units during $[c,d)$.
		As a result, the job finishes no later than at time $d$.
	\end{proof}
	
	\begin{lem}\label{lem:non-sched}
		Consider some interval $[c,d)$ with $c\leq d$.
		If $\tau_{k,\ell}$ is not finished by time d and the task $\tau_k$ is active during (the whole interval) $[c,d)$, then 
		\begin{equation}
			(C_k+S_k) + \sum_{i \neq k} B^i_{k,\ell}(c,d) + \sum_{j < \ell} B_{k,j}(c,d) > (d-c).
		\end{equation}
	\end{lem}

	\begin{proof}
		Since $\tau_k$ is active during $[c,d)$, we have $\Btilde_k(c,d) = 0$.
		Using this together with Lemma~\ref{lem:distribution_ps_with_b} and Lemma~\ref{lem:refinement_Bk} leads to $\sum_{i\neq k} B^i_k(c,d) + \sum_{j\in\Nbb} B_{k,j}(c,d) = (d-c)$.
%		\begin{equation}
%			 \sum_{i\neq k} B^i_k(c,d) + \sum_{j\in\Nbb} B_{k,j}(c,d) = (d-c).
%		\end{equation}
		We prove the lemma by showing that 
%		1) $\sum_{j \in \Nbb} B_{k,j}(c,d) < C_k+S_k + \sum_{j < \ell} B_{k,j}(c,d)$ and 2) $B^i_k(c,d)\leq B^i_{k,\ell}(c,d)$ for all $i\neq k$.
		\begin{enumerate}
			\item $\sum_{j \in \Nbb} B_{k,j}(c,d) < C_k+S_k + \sum_{j < \ell} B_{k,j}(c,d)$ and
			\item $B^i_k(c,d)\leq B^i_{k,\ell}(c,d)$ for all $i\neq k$.
		\end{enumerate}
		
		First, we prove 1). 
		Since $\tau_{k,\ell}$ is not finished by time $d$, we know that all $\tau_{k,j}$ with $j>\ell$ are not current before $d$.
		Hence, $B_{k,j}(c,d) = 0$ for all $j > \ell$ and $\sum_{j \in \Nbb} B_{k,j}(c,d) = \sum_{j \leq \ell} B_{k,j}(c,d)$.
		Moreover, due to Lemma~\ref{lem:Bkl_finished}, we have $B_{k,\ell}(c,d) < C_k+S_k$.
		This proves 1).
		
		Second, we prove 2).
		Since $\tau_{k}$ is active during $[c,d)$, by definition, at any time instant $t \in [c, d)$, there is always a job of $\tau_k$, which is current. Let the last job of $\tau_k$ that is current in this interval be $\tau_{k,\ell'}$. That is, there exists some $\eps >0$ such that $\tau_{k,\ell'}$ is current during $[d-\eps,d)$.  By Lemma~\ref{lem:refinement_Bik}, we have $B^i_k(c,d) \leq B^i_{k,\ell'}(c,d)$ for all $i\neq k$.

        According to the handling mechanism to deal with multiple
        active jobs of $\tau_k$ (for arbitrary-deadline task systems) at the end of Section~\ref{sec:system_model}, the jobs of
        $\tau_k$ are executed one after another. Since
        $\tau_{k,\ell'}$ is current in $[d-\eps,d)$, the job
        $\tau_{k,\ell}$ can not have higher priority than
        $\tau_{k, \ell'}$, i.e., $\ell' \leq \ell$ holds.
        Otherwise,
        $\tau_{k,\ell}$ has to be finished by time $d-\eps$.
        Therefore, $B^i_{k,\ell'}(c,d) \leq B^i_{k,\ell}(c,d)$ for all
        $i\neq k$. 
        We conclude that 2) holds.
	\end{proof}
	
	The negation of this lemma leads to a response-time upper bound of $\tau_{k,\ell}$ introduced by Proposition~\ref{prop:backbone_prep} and Corollary~\ref{cor:backbone}.
	
	\begin{prop}\label{prop:backbone_prep}
      Let $c\leq d \in \Rbb$ such that either 1) $\tau_{k,\ell}$ is
      released before $c$, i.e., $c \geq r_{k,\ell}$, or 2)~$c < r_{k,\ell}$ and task $\tau_k$ is active during
        $[c,r_{k,\ell})$. If
		\begin{equation}\label{eq:prop:backbone_prep}
			(C_k+ S_k) + \sum_{i \neq k} B^i_{k,\ell}(c,d) + \sum_{j < \ell} B_{k,j}(c,d) \leq (d-c),
		\end{equation}
		then $\tau_{k,\ell}$ is finished by time $d$, i.e., $d \geq f_{k,\ell}$.
	\end{prop}
	
	\begin{proof}
		We use an indirect proof strategy and assume that $\tau_{k,\ell}$ is not finished by time $d$.
		In this case, $\tau_{k,\ell}$ is active during $[r_{k,\ell},d)$.
		Since $r_{k,\ell}\leq c$ or $\tau_k$ is active during $[c,r_{k,\ell})$, we conclude that $\tau_k$ is also active during $[c,d)$.
		Lemma~\ref{lem:non-sched} states that Equation~\eqref{eq:prop:backbone_prep} does not hold.
		This contradicts our assumption.
	\end{proof}
	
	In Section~\ref{sec:estimation_B} we obtain the upper bounds for $B^i_{k,\ell}(c,d)$ and $B_{k,j}(c,d)$ by enlarging the window of interest to $[c, d_{k,\ell})$.
	To apply these upper bounds we formulate the response-time upper bound using $B^i_{k,\ell}(c, d_{k,\ell})$ and $B_{k,j}(c, d_{k,\ell})$ in the following corollary.

	\begin{cor}[Analysis Backbone]\label{cor:backbone}
    Let $c\leq d_{k,\ell} \in \Rbb$ such that either 1)~$\tau_{k,\ell}$ is
    released before $c$, i.e., $c \geq r_{k,\ell}$ or 2)~$c < r_{k,\ell}$ and task $\tau_k$ is active during
        $[c,r_{k,\ell})$.
		If 
		\begin{equation}\label{eq:cor_sched_test}
			\Rtilde_{k,\ell} \!:=\!
			(C_k+ S_k) + \!\sum_{i \neq k}\! B^i_{k,\ell}(c,d_{k,\ell}) + \!\sum_{j < \ell}\! B_{k,j}(c,d_{k,\ell}) + c - r_{k,\ell},
		\end{equation}
		is at most $D_k$, then $\Rtilde_{k,\ell}$ is an upper bound on the response time of $\tau_{k,\ell}$.
	\end{cor}
	
	\begin{proof}
		This corollary follows from Proposition~\ref{prop:backbone_prep} by setting $d:=r_{k,\ell} + \Rtilde_{k,\ell}$.
		The case $c>d$ does not occur since otherwise $\Rtilde_{k,\ell} \geq  c-r_{k,\ell} > \Rtilde_{k,\ell}$ would hold by definition of $\Rtilde_{k,\ell}$ and $d$.
		We only consider $c \leq d$.
		
		Since $\Rtilde_{k,\ell} \leq D_k$ by assumption, we have $d\leq d_{k,\ell}$.
		Hence, $B^i_{k,\ell}(c,d)\leq B^i_{k,\ell}(c,d_{k,\ell})$ and $B_{k,j}(c,d) \leq B_{k,j}(c,d_{k,\ell})$ for all $i$ and all $j$.
		As a result, the left hand side of Equation~\eqref{eq:prop:backbone_prep} is less than or equal to $(C_k+ S_k) + \sum_{i \neq k} B^i_{k,\ell}(c,d_{k,\ell}) + \sum_{j < \ell} B_{k,j}(c,d_{k,\ell})$, which is the same as $\Rtilde_{k,\ell} + r_{k,\ell} - c = d-c$ by definition.
		
		Since all assumptions from Proposition~\ref{prop:backbone_prep} hold, the job $\tau_{k,\ell}$ is finished by time $d=\Rtilde_{k,\ell} + r_{k,\ell}$.
		Hence, $\Rtilde_{k,\ell}$ is an upper bound on the response time of $\tau_{k,\ell}$.
	\end{proof}

	To apply the response time upper bound from Equation~\eqref{eq:cor_sched_test} we need to answer the following two questions:
	\begin{itemize}
		\item \textbf{Question~1:} What are the values of $B^i_{k,\ell}(c,d_{k,\ell})$ and $B_{k,j}(c,d_{k,\ell})$?
		Since computing the values directly has high complexity, we use overapproximation to deal with this issue.
		In Section~\ref{sec:estimation_B} we derive upper bounds for $B^i_{k,\ell}(c,d_{k,\ell})$ and $B_{k,j}(c,d_{k,\ell})$ with $i\neq k$ and $j<\ell$.
		\item \textbf{Question~2:} Which are good values for $c$?
		Trying out all possible $c$ for the estimation would lead to very high complexity of our method.
		Therefore, we discuss two strategies to choose $c$ in Sections~\ref{sec:test_fixed} and~\ref{sec:test_arb}.
		More precisely, with the first procedure we restrict $c$ to be in the interval $[r_{k\ell}, d_{k,\ell})$, which has benefits on the runtime of our analysis due to the \emph{fixed analysis windows}.
		For the second strategy, we examine \emph{active intervals} for $\tau_k$, and gradually increase the analysis window.
		In Section~\ref{sec:dominance} we discuss that both methods do not dominate each other.
	\end{itemize}
	
	\subsection{Upper Bounds of $B^i_{k,\ell}(c,d_{k,\ell})$ and $B_{k,j}(c,d_{k,\ell})$}
	\label{sec:estimation_B}
	
	In this section, we bound the interference of higher-priority jobs during the interval $[c,d_{k,\ell})$ to provide upper bounds
	for $\sum_{j < \ell} B_{k,j}(c,d_{k,\ell})$ in Lemma~\ref{lem:upper_bound_sum_Bkj} and for $B_{k,\ell}^i(c, d_{k,\ell})$ in Lemma~\ref{lem:total_bound_interference}.
	We do this under the assumption that all jobs with higher priority than $\tau_{k,\ell}$ meet their deadline, since this is our induction hypothesis later used in Theorem~\ref{thm:suff_sched_test_fixed} and Lemma~\ref{lem:suff_sched_test_var_prep}.
	As mentioned earlier, the decision of $c$ will be further discussed in Sections~\ref{sec:test_fixed}~and~\ref{sec:test_arb}.
	For each task $\tau_i \neq \tau_k$ let $\Rtilde_i$ be an upper bound on the worst-case response time (WCRT) of jobs of $\tau_i$ with higher priority than $\tau_{k,\ell}$.
	If there is no upper bound available yet, we set $\Rtilde_i := D_i$, as all jobs with higher priority than $\tau_{k,\ell}$ meet their deadline.
    
	With arbitrary deadlines, there might be several active jobs of one task at the same time, which makes the analysis in general more complicated.
	Note that, according to the mechanism introduced in the end of Section~\ref{sec:system_model}, the jobs of $\tau_k$ must be executed one after another, i.e., even if the processor idles, a job of $\tau_k$ cannot start its execution unless all jobs of $\tau_k$ released prior to it are finished.
	
	\begin{lem}\label{lem:number_active_jobs}
		The number of jobs of $\tau_k$ with higher priority than $\tau_{k,\ell}$ which can be executed during $[c,d_{k,\ell})$ is upper bounded by $\ceiling{\frac{d_{k,\ell}-c}{T_k}}-1$.
	\end{lem}
	
	\begin{proof}
		The formula $\ceiling{\frac{d_{k,\ell}-c}{T_k}}$ counts the maximal number of deadlines of $\tau_k$ in $(c,d_{k,\ell}]$. Since $\tau_{k,\ell}$ has deadline at the end of the analysis window but it has no higher priority than itself, we remove one job.
	\end{proof}

	We use the previous lemma to conclude an upper bound for $\sum_{j < \ell} B_{k,j}(c,d_{k,\ell})$ as follows.

	\begin{lem}[Bound~B1]\label{lem:upper_bound_sum_Bkj}
		The total amount of time that the processor is working on or suspended by higher-priority jobs of $\tau_k$ while not working on higher priority jobs of other tasks is upper bounded by
		\begin{equation}
			\sum_{j < \ell} B_{k,j}(c,d_{k,\ell}) \leq \left(\ceiling{\frac{d_{k,\ell}-c}{T_k}}-1\right) \cdot (C_k+S_k).
		\end{equation}
	\end{lem}
	\begin{proof}
		This follows from Lemma~\ref{lem:number_active_jobs} and the upper bounds $C_k$ for execution and $S_k$ for suspension. 
	\end{proof}

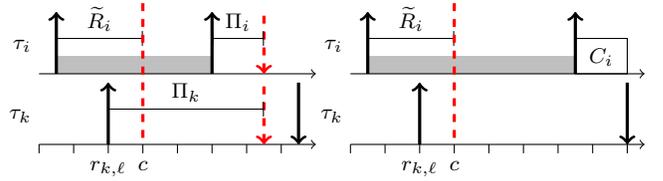
\begin{figure}
	\centering
	\begin{tikzpicture}[yscale=\figurestretchvert, xscale=0.23]
		\footnotesize{}   
		
%		\begin{scope}[shift={(-14,1)}]
%			\taskname{$\tau_{i,j}$}
%			
%			\timeline{-1}{10}{}
%			\grid{0}{10}{1}{1}
%			
%			\releases{0}
%			\deadlines{9}
%			\glprio{7}
%			
%			\exec{0}{2}
%			\exec{4}{5}
%			
%			\susp{2}{4}
%			
%			\node[below] at (0,-.25) {$r_{i,j}$};
%			\node[below] at (5,-.25) {$f_{i,j}$};
%			\node[below] at (7,-.25) {\color{red} $\pi_{i,j}$};
%			\node[below] at (9,-.25) {$d_{i,j}$};
%			
%			\draw[|-|] (0,2) -- node[above] {$\Pi_i$} (7,2);
%			\draw[|-|] (0,-1.2) -- node[below] {$D_i$} (9,-1.2);
%		\end{scope} 
		
		\begin{scope}[shift={(0,2)}] % task one
			
			\fill[black!25] (1,0) rectangle (10, 0.5);
			
			\taskname{$\tau_i$}
			
			\timeline{0}{16}{}
			% no labelling
			
			\releases{1,10}
%			\deadlines{5,10,15}
			
%			\exec{0}{2}
%			\exec{5}{7}
%			\exec{13}{15}
			
			\glprio{13}
			
			\draw[|-|] (1,1) -- node[above] {$\Rtilde_i$} (6,1);
			\draw[|-|] (10,1) -- node[above] {$\Pi_i$} (13,1);

			% no suspension
		\end{scope}
		
		\begin{scope}[shift={(0,0)}] % task two
			\taskname{$\tau_k$}
			
			\timeline{0}{16}{}
			\grid{0}{15}{2}{0}
%			\labelling{0}{17}{2}{0}
			
			\releases{4}
			\deadlines{15}
			
%			\exec{2}{3}
%			\exec{7}{13}
			
%			\susp{3}{6}
			
			\glprio{13}
			
			\node[below] at (4,-.25) {$r_{k,\ell}$};
			\node[below] at (6,-.25) {$c$};
			
			\draw[|-|] (4,1) -- node[above] {$\Pi_k$} (13,1);
		\end{scope}
	
		\draw[color=red, dashed, very thick] (6,4) -- (6,0);

		\begin{scope}[shift={(18,0)}]
			\begin{scope}[shift={(0,2)}] % task one
				
				\fill[black!25] (1,0) rectangle (13, 0.5);
				
				\taskname{$\tau_i$}
				
				\timeline{0}{17}{}
				% no labelling
				
				\releases{1,13}
				%			\deadlines{5,10,15}
				
				\exec{13}{16}
				\node[above] at (14.5,0) {$C_i$}; 
				%			\exec{5}{7}
				%			\exec{13}{15}
				
%				\glprio{14}
				
				\draw[|-|] (1,1) -- node[above] {$\Rtilde_i$} (6,1);
%				\draw[|-|] (11,1) -- node[above] {$\Pi_i$} (14,1);

				% no suspension
			\end{scope}
			
			\begin{scope}[shift={(0,0)}] % task two
				\taskname{$\tau_k$}
				
				\timeline{0}{17}{}
				\grid{0}{17}{2}{0}
				%			\labelling{0}{17}{2}{0}
				
				\releases{4}
				\deadlines{16}
				
				%			\exec{2}{3}
				%			\exec{7}{13}
				
				%			\susp{3}{6}
				
%				\glprio{14}
				
				\node[below] at (4,-.25) {$r_{k,\ell}$};
				\node[below] at (6,-.25) {$c$};
				
%				\draw[|-|] (4,1) -- node[above] {$\Pi_k$} (14,1);
			\end{scope}
			
			\draw[color=red, dashed, very thick] (6,4) -- (6,0);
		\end{scope}

	\end{tikzpicture}
	\caption{Intuition for Lemma~\ref{lem:ineqk_1}~(left) and Lemma~\ref{lem:ineqk_2}~(right).
	The gray box depicts the interval to count job releases.}\label{fig:intuition}
\end{figure}

	For $B^i_{k,\ell}$, we bound the interference by estimating the number of job releases during a certain interval under analysis, as depicted by the gray boxes in Figure~\ref{fig:intuition}.

	\begin{lem}\label{lem:ineqk_1}
		There are at most $\max\left(\ceiling{\frac{\Pi_k-\Pi_i+\Rtilde_i+r_{k,\ell}-c}{T_i}}, 0\right)$ jobs of $\tau_i,~i\neq k$ with higher priority than $\tau_{k,\ell}$ that are executed during the interval $[c, d_{k,\ell})$.
	\end{lem}

	\begin{proof}
		A job of $\tau_i$ can only be executed during $[c,d_{k,\ell})$, if it is released after $c-\Rtilde_i$.
		Moreover, it can have higher priority than $\tau_{k,\ell}$, if it is released no later than $r_{k,\ell} + \Pi_k - \Pi_i$.
		We count the maximal number of releases of $\tau_i$ during $(c-\Rtilde_i, r_{k,\ell} + \Pi_k - \Pi_i)$ by $\ceiling{\frac{r_{k,\ell} + \Pi_k - \Pi_i - c + \Rtilde_i}{T_i}}$ if the interval has non-negative length.
		Otherwise, the number of releases is $0$.
	\end{proof}
	
	The processor is working on each job of $\tau_i$ for at most $C_i$ time units.
	Hence, we have an upper bound of $B_{k,\ell}^i(c,d_{k,\ell}) \leq \max\left(\ceiling{\frac{\Pi_k-\Pi_i+\Rtilde_i+r_{k,\ell}-c}{T_i}}, 0\right) \cdot C_i$.
	In the above proof we only compare if the job can work \emph{after} $c$.
	We do not check if the job can work \emph{before} $d_{k,\ell}$.
	As a result, the estimation can become very loose if $\Pi_k$ is very high compared to $\Pi_i$, e.g., when applying the schedulability test to TFP scheduling.
	Therefore, we combine the estimation with a second approach.
	
	\begin{lem}\label{lem:ineqk_2}
		For $i\neq k$,  $B_{k,\ell}^i(c,d_{k,\ell})$ is upper bounded
        by
        $\max\left\{\ceiling{\frac{D_k-C_i+r_{k,\ell}-c+\Rtilde_i}{T_i}}C_i, 0\right\}$.
	\end{lem}
	
	Intuitively, the maximal interference from task $\tau_i$ is obtained when the \emph{last} interfering job of $\tau_i$ is released at $d_{k,\ell}-C_i$ and executed for $C_i$ time units during $[d_{k,\ell}-C_i, d_{k,\ell})$.
	The maximum interference is then calculated by the number of releases during $(c-\Rtilde_i, d_{k,\ell}-C_i]$ multiplied with $C_i$.
	In the following, we provide a formal proof as well.

	\begin{proof}
		If there is no job of $\tau_i$ executed during $[c,d_{k,\ell})$, then $B^i_{k,\ell}(c,d_{k,\ell})=0$ and the lemma is proven.
		
		Otherwise, let $\tau_{i,j'}$ be the last job of $\tau_i$ being executed during $[c,d_{k,\ell})$.
		We isolate that job in the following way.
		Let $r_{i,j'}$ be the release time of $\tau_{i,j'}$ and let $C^*$ be the amount of time that the processor is working on $\tau_{i,j'}$ during the interval $[c,d_{k,\ell})$.
		We have $B(c,d_{k,\ell}) \leq B(c,r_{i,j'}) + C^*.$
%		\begin{equation}
%			B(c,d_{k,\ell}) \leq B(c,r_{i,j'}) + C^*.
%		\end{equation}
		
		If $r_{i,j'}-c < (C-C^*)$, then $B(c,r_{i,j'}) + C^* \leq C$.
		Moreover, $\ceiling{\frac{D_k-C_i+r_{k,\ell}-c+\Rtilde_i}{T_i}}C_i \geq \ceiling{\frac{D_k+r_{k,\ell}-c}{T_i}}C_i \geq C_i$ since there is a job being executed during $[c, d_{k,\ell})$, i.e., $c < D_k+r_{k,\ell}$.
		In this case the lemma is proven.
		
		If $r_{i,j'}-c \geq (C-C^*)$, then $B(c,r_{i,j'}) + C^* \leq B(c + (C_i-C^*),r_{i,j'}) + C_i$, i.e., we over approximate be replacing $(C_i-C^*)$ time units at the beginning by execution time.
		Next we compute $B(c + (C_i-C^*),r_{i,j'})$ by counting the number of releases during the interval $(c+(C_i-C^*)-\Rtilde_i, r_{i,j'}-T_i]$.
		It is 
%		\begin{equation}
		$
			\ceiling{\frac{r_{i,j'}-T_i-c-C_i+C*+\Rtilde_i}{T_i}} 
			\leq \ceiling{\frac{r_{k,\ell}-T_i-c-C_i+D_k+\Rtilde_i}{T_i}}
			= \ceiling{\frac{r_{k,\ell}-c-C_i+D_k+\Rtilde_i}{T_i}} -1
		$
%		\end{equation}
		where we use that $r_{i,j'} + C^* \leq r_{k,\ell} + D_k$.
		In total, the value of $B(c,d_{k,\ell})$ is then upper bounded by $\ceiling{\frac{r_{k,\ell}-c-C_i+D_k+\Rtilde_i}{T_i}} \cdot C_i$ in this case.
	\end{proof}
	
	We combine the above two lemmas, to obtain the following bound for $B_{k,\ell}^i(c,d_{k,\ell})$.
	
	\begin{lem}[Bound~B2]\label{lem:total_bound_interference}
		For $i\neq k$ the value of  $B_{k,\ell}^i(c,d_{k,\ell})$ is upper bounded by
		\begin{equation}\label{eq:total_bound_interference}
			B_{k,\ell}^i (c,d_{k,\ell}) \leq
				\max\left(\ceiling{\frac{G_{k}^i+\Rtilde_i+r_{k,\ell}-c}{T_i}}, 0\right)C_i
		\end{equation}
		where $G_k^i:= \min(D_k-C_i, \Pi_k-\Pi_i)$.
	\end{lem}

	\begin{proof}
		This is a combination of the results from Lemma~\ref{lem:ineqk_1} and Lemma~\ref{lem:ineqk_2}.
	\end{proof}

	\subsection{Fixed Analysis Window}
\label{sec:test_fixed}	
	\begin{algorithm}[t]
		\small
		%\footnotesize
		%\scriptsize
		\caption{Schedulability test with fixed analysis window.}
		\label{alg:sched_test_fixed}
		\begin{algorithmic}[1]
			\item[] \textbf{Input:} $\Tbb = \{\tau_1,\dots,\tau_n\}$, $(\Pi_1,\dots,\Pi_n)$, $\eta$, $depth$
			\item[] \textbf{Output:} True: schedulable, False: no decision
			\item[]
			
			\State Order $\tau_1 ,\dots,\tau_n$, s.th. $D_1\geq \dots \geq D_n$.
			\State Set $\Rtilde_i := D_i$ for all $i$.
			
			\For{$i=1, 2, \dots, depth$}
			\State $solved := True$
			\For{$k = 1, 2, \dots,n$}
			\State $cand := [~]$; $step := \eta \cdot {D_k}$
			\Comment{Preparation.}
			\For{$b = 0, step, 2\cdot step, \cdots < D_k$}
			\Comment{Compute.}
%			\Comment{Compute candidate.}
			\State $cand.append(\Rtilde_k(b))$ using Equation~\eqref{eq:suff_sched_test_fixed}.
			\EndFor
			\State $\Rtilde_k := \min\limits (cand)$
			\Comment{Compare candidates.}
			\If{$\Rtilde_k>D_k$} 
			\Comment{Check condition.}
%			\Comment{Check schedulability condition.}
			\State $solved := False$; $\Rtilde_k := D_k$; \textbf{break}
			\EndIf
			\EndFor
			\EndFor
			\State \Return $solved$
		\end{algorithmic}
	\end{algorithm}
	
	In this section, we fix the analysis window, i.e., the possible range of $c$ from the previous sections, to the interval $[r_{k,\ell}, d_{k,\ell})$.
	First, we utilize the upper bounds on $B^i_{k,\ell}(c,d_{k,\ell})$ and $B_{k,j}(c,d_{k,\ell})$ provided in the previous section, to obtain the following schedulability test for this scenario with fixed analysis interval.
	
	\begin{thm}[Sufficient Schedulability Test]\label{thm:suff_sched_test_fixed}
		Let $\Tbb = \{\tau_1,\dots.\tau_n\}$ be an arbitrary-deadline task set with relative priority points $(\Pi_1, \dots, \Pi_n)$.
		If for all $k=1, \dots, n$ there exists some $b_k \in [0,D_k)$ such that
		\begin{equation}\label{eq:suff_sched_test_fixed}
%			\footnotesize{}
			\Rtilde_k(b_k)
%			:=\ceiling{\frac{D_k-b_k}{T_k}}(C_k + S_k) + b_k + \sum\limits_{i\neq k} \max\left(\ceiling{\frac{G_k^i+ \Rtilde_i-b_k}{T_i}}, 0\right) C_i 
			\leq D_k,
		\end{equation}
		where $\Rtilde_k(b_k):=\ceiling{\frac{D_k-b_k}{T_k}}(C_k + S_k) + b_k + \sum\limits_{i\neq k} \max\left(\ceiling{\frac{G_k^i+ \Rtilde_i-b_k}{T_i}}, 0\right) C_i$
		and 
%		where $\Rtilde_k(b_k)$ is computed by  $\ceiling{\frac{D_k-b_k}{T_k}}(C_k + S_k) + b_k + \sum\limits_{i\neq k} \max\left(\ceiling{\frac{G_k^i+ \Rtilde_i-b_k}{T_i}}, 0\right) C_i$ and 
%		where
		$G^i_k = \min(D_k-C_i, \Pi_k-\Pi_i)$,
		then the task set is schedulable by EL scheduling with the given relative priority points and the worst-case response time of $\tau_k$ is upper bounded by $\Rtilde_k := \Rtilde_k(b_k)$.
	\end{thm}
	
	\begin{proof}
		Assume we have found $b_k,~k=1,\dots,n$ such that Equation~\eqref{eq:suff_sched_test_fixed} holds.
		We consider some schedule obtained by this task set and denote by $Seq$ the sequence of all jobs in the schedule ordered by their priority.
		Via induction, we prove that the first $\xi$ jobs in $Seq$ have the required response time upper bound, for all $\xi \in \Nbb_0$.
		Consequently, $\Rtilde_k$ is an upper bound on the worst-case response time of $\tau_k$ for all $k$ and the task set is schedulable.
		
		\textbf{Initial case:} $\xi = 0$.
		In this case, the set of the first $\xi$ jobs in $Seq$ is $\{\}$ the empty set.
		Trivially, all of them have the required response time upper bound.
		
		\textbf{Induction step:} $\xi \to \xi+1$.
		By assumption, the first $\xi$ jobs in $Seq$ have the required response time upper bound.
		We denote the $(\xi+1)$-th job in $Seq$ by $\tau_{k,\ell}$.
		We aim to use Corollary~\ref{cor:backbone} to prove that the response time of $\tau_{k,\ell}$ is upper bounded by $\Rtilde_k$.
		By definition, we have $\Rtilde_{k,\ell} = (C_k+ S_k) + \sum_{i \neq k} B^i_{k,\ell}(c,d_{k,\ell}) + \sum_{j < \ell} B_{k,j}(c,d_{k,\ell}) + c - r_{k,\ell}$.
		Since all higher priority jobs have the required response time upper bound, we can use the estimation from Section~\ref{sec:estimation_B}.
		Using Lemma~\ref{lem:upper_bound_sum_Bkj} and Lemma~\ref{lem:total_bound_interference}, we obtain $\Rtilde_{k,\ell} \leq \Rtilde_k$ when choosing $c:= b_k+r_{k,\ell}$.
		Due to Equation~\eqref{eq:suff_sched_test_fixed}, even
		$\Rtilde_{k,\ell} \leq \Rtilde_k \leq D_k$.
%		\begin{equation}
%			\Rtilde_{k,\ell} \leq \Rtilde_k \leq D_k.
%		\end{equation}
		We use Corollary~\ref{cor:backbone} to conclude that $\Rtilde_{k,\ell}$ is an upper bound on the response time of $\tau_{k,\ell}$ and therefore $\Rtilde_k$ is an upper bound on the response time of $\tau_{k,\ell}$ as well.
	\end{proof}
    Although the test in Theorem~\ref{thm:suff_sched_test_fixed} looks like a classical mechanism extended from time demand analysis (TDA)~\cite{joseph86responsetimes,lehoczky89}, implementing an efficient schedulability test based on it requires some efforts since the values of $\Rtilde_k$ for every task $\tau_k$ are dependent on each other. The critical part to apply this schedulability is two folded:
	\begin{enumerate}
		\item Find good values for $b_k$ with low complexity. Without an efficient mechanism, there are $D_k$ options for $b_k$, provided that all input parameters are integers, and infinitely many options in general.
		\item Compute the dependent values of $\Rtilde_k$ for every task $\tau_k$ correctly and efficiently.
	\end{enumerate}
    For $b_k$, we discretize the search space into $\frac{1}{\eta}$ values with a step size $\eta \cdot D_k$, for a user-specified parameter $\eta$. For computing $\Rtilde_k$, we go through the task set several times to compute upper bounds for the values of $\Rtilde_k$.
	Improvement of the search algorithm will be part of future work and is out of scope for this paper.

	The search algorithm is depicted in pseudocode in Algorithm~\ref{alg:sched_test_fixed}.
	It takes as input the task set $\Tbb$, the relative priority points $(\Pi_1, \dots, \Pi_n)$, a step size parameter $\eta \in (0,1]$ and $depth$ to indicate the number improving runs of the search algorithm.
	It returns \emph{True} if the task set is schedulable by EL scheduling with the given relative priority points.
	We start by setting $\Rtilde_k = D_k$ for all $k=1, \dots, n$, and go $depth$-times through the task set ordered by the relative deadline, as we obtained the best results with this ordering.
	With  a step size of $step = \eta \cdot D_k$, i.e., a certain share of $D_k$ like $1$ percent, we compute $\Rtilde_k(b_k)$ for $b_k=0,step,2\cdot step, 3\cdot step, \dots$ until $b \geq D_k$ is reached.
	We then take the minimal value of all these candidates and define it as the new $\Rtilde_k$.
    The time complexity of Algorithm~\ref{alg:sched_test_fixed} is $\mathcal{O}\left(\frac{depth\cdot n^2}{\eta}\right)$.
	
	Please note that the computed values of $\Rtilde_k$ are in fact only upper bounds of $\Rtilde_k$ from Theorem~\ref{thm:suff_sched_test_fixed}.
	A reduction of $\Rtilde_i,~i\neq k$ in subsequent iterations reduces the actual value of $\Rtilde_k$ as well, since $\Rtilde_k$ is monotonically increasing with respect to $\Rtilde_i,$ for all $i\neq k$.

\subsection{Variable Analysis Window}
	\label{sec:test_arb}

	In the following, we show a different approach based on \emph{active intervals}.
	More specifically, if all jobs finish until the next job release is reached, i.e., $R_k \leq T_k$, then no previous jobs contribute interference to the job under analysis and they can be safely removed from the computation of the worst-case response-time upper bound.
	However, if $R_k > T_k$ then interference from previous jobs has to be considered in the following way.
	We utilize that a job $\tau_{k,\ell-a}$ can only interfere with $\tau_{k,\ell}$, if $\tau_k$ is active during $[r_{k,\ell-a}, r_{k,\ell})$.
	For the schedulability test with variable analysis window, we gradually increase the length of the active interval, i.e., $a = 0,1,2,\dots$ and analyze the window $[c,d_{k,\ell})$ with $c \in [r_{k,\ell-a}, d_{k,\ell})$.
	
	With this approach, the pessimism of the interference estimation from higher-priority jobs of the same task is reduced in some cases.
	Please note that this approach only differs from the case with fixed analysis window when considering arbitrary deadline tasks.
	For constrained deadline task sets, the variable analysis window approach coincides with the fixed analysis window approach, as the algorithm stops at $a=0$ without enlarging the analysis window.
	We start be formally defining active intervals.
	
	\begin{defn}\label{def:active_int}
		Let $a\in \Nbb_0$.
		A job $\tau_{k,\ell}$ is the $(a+1)$-th job in an \emph{active interval} of $\tau_k$, if the following two conditions hold.
		\begin{itemize}
			\item $\tau_k$ is active during $[r_{k,\ell-a},f_{k,\ell})$.
			\item At time $r_{k,\ell-a}$ there is no active job which is released \emph{before} $r_{k,\ell-a}$.
		\end{itemize}
	\end{defn}

	If $\tau_{k,\ell}$ is the $(a+1)$-th job in an active interval of $\tau_k$, then only $\tau_{k,\ell-a}, \dots, \tau_{k,\ell}$ are current jobs of $\tau_k$ during $[r_{k,\ell-a},f_{k,\ell})$.
	More specifically, in this case the value of $B_{k,j}(c,d_{k,\ell})$ is $0$ if $c \geq r_{k,\ell-a}$ and $j<\ell-a$. 
	We formalize this by the following lemma.
	
	\begin{lem}\label{lem:suff_sched_test_var_prep}
		Let $\tau_{k,\ell}$ be the $(a+1)$-th job in an active interval of $\tau_k$ and let all higher-priority jobs meet their deadline.
		Let $\Rtilde_i,~i\neq k$ be an upper bound on the response time of all higher-priority jobs of $\tau_i$.
		If there exists some $c \in [r_{k,\ell-a},d_{k,\ell})$ such that 
		\begin{equation}\label{eq:lem_var_window}
			\begin{aligned}
				&\min \left( a+1, \ceiling{\frac{d_{k,\ell}-c}{T_k}} \right) (C_k+S_k)
				\\&+ \sum_{i \neq k} \max\left( \ceiling{\frac{G^i_k+\Rtilde_i+r_{k,\ell}-c}{T_i}}, 0 \right) C_i
				+ c - r_{k,\ell}
			\end{aligned}
		\end{equation}
		is at most $D_k$,
		with $G^i_k := \min(D_k-C_i, \Pi_k-\Pi_i)$, then \eqref{eq:lem_var_window} is an upper bound on the response time of $\tau_{k,\ell}$.
	\end{lem}
	\begin{proof}
		For the proof, we apply Corollary~\ref{cor:backbone}.
		Since $\tau_{k,\ell}$ is the $(a+1)$-th job in an active interval, $\tau_k$ is active during $[r_{k,\ell-a}, r_{k,\ell})$.
		Hence, the restriction on $c$ in Corollary~\ref{cor:backbone} are fulfilled by default, when $c$ is chosen from $[r_{k,\ell-a}, d_{k,\ell})$.
		Moreover, since $\tau_{k,\ell}$ is the $(a+1)$-th job in an active interval, the jobs $\tau_{k,1}, \dots, \tau_{k,\ell-a-1}$ are finished by time $r_{k,\ell-a}$.
		We obtain $\sum_{j < \ell-a} B_{k,j} (c,d_{k,\ell}) \leq \sum_{j < \ell-a} B_{k,j} (r_{k,\ell-a},d_{k,\ell}) = 0$.
		Hence, 
%		\begin{equation}
		$
			\sum_{j < \ell} B_{k,j} (c,d_{k,\ell}) 
			= \sum_{j = \ell-a}^{\ell-1} B_{k,j} (c,d_{k,\ell})
			\leq \sum_{j = \ell-a}^{\ell-1} (C_k+S_k)
			= a\cdot (C_k+S_k).
%		\end{equation}
		$
		We combine this with the results from Lemma~\ref{lem:upper_bound_sum_Bkj} and Lemma~\ref{lem:total_bound_interference}, and obtain that $\Rtilde_{k,\ell}$ from Equation~\eqref{eq:cor_sched_test} is less than or equal to the value in Equation~\eqref{eq:lem_var_window}.
		If \eqref{eq:lem_var_window} is at most $D_k$, then $\Rtilde_{k,\ell}\leq D_k$.
		Corollary~\ref{cor:backbone} states that $\Rtilde_{k,\ell}$ is an upper bound on the response time of $r_{k,\ell}$ and therefore, also \eqref{eq:lem_var_window} is an upper bound on the response time.
	\end{proof}
	
	In the following theorem, we replace $r_{k,\ell}-c$ by $aT_k-x$.

	\begin{thm}[Sufficient Schedulability Test]\label{thm:suff_sched_test_var}
		Let $\Tbb=\{\tau_1, \dots, \tau_n\}$ be an arbitrary-deadline task set with relative priority points $(\Pi_1, \dots, \Pi_n)$.
		We define the function $\Rtilde^a_k : \Rbb_{\geq 0} \to \Rbb_{\geq 0}$ by the assignment 
		\begin{equation}\label{eq:thm:suff_sched_test_var}
		\begin{aligned}
		x \mapsto &\min\left(a+1, \ceiling{\frac{D_k-x+aT_k}{T_k}}\right)(C_k+S_k)
		\\&+ \sum_{i \neq k} \max \left( \ceiling{\frac{G^k_i + \Rtilde_i-x+aT_k}{T_i}},0 \right) C_i
		+ x-aT_k.
		\end{aligned}
		\end{equation}
		If for all $k=1,\dots,n$ there exists $\tilde a_k \in \Nbb_0$, such that for all $a = 0,\dots,\tilde a_k$ there exists $b^a_k \in [0,a T_k + D_k)$, such that 
		\begin{equation}
			\Rtilde^a_k(b^a_k) \leq D_k,
			\text{ and furthermore }
			\Rtilde^{\tilde a_k}_k(b^{\tilde a_k}_k) \leq T_k,
		\end{equation}
		then the task set is schedulable by EL scheduling with the given relative priority points and $\Rtilde_k := \max_{a = 0, \dots, \tilde a} \Rtilde^a_k(b^a_k)$ is an upper bound on the WCRT of $\tau_k$ for all $k$. 
	\end{thm}
	
	\begin{proof}
		The proof is similar to the one of Theorem~\ref{thm:suff_sched_test_fixed}.
		Let $Seq$ be the sequence of all jobs in the schedule ordered by their priority.
		By induction we show that the following response time upper bound holds for the first $\xi \in \Nbb_0$ jobs in $Seq$:
		\begin{enumerate}
			\item $\Rtilde_k$ for all jobs of $\tau_k$ for each task $\tau_k$.
			\item $T_k$ for all $\tilde a_k$-th jobs in an active interval of $\tau_k$ for each task $\tau_k$.
		\end{enumerate}
		
		\textbf{Initial case:} $\xi_0$. The initial case is again trivially fulfilled, since there has nothing to be checked when there are no jobs.
		
		\textbf{Induction step:} $\xi \to \xi+1$.
		The first $\xi$ jobs in $Seq$ have the required response time upper bounds 1) and 2) by induction.
		We denote by $\tau_{k,\ell}$ the $(\xi+1)$-th job of $Seq$.
		Let $a$ be the lowest value in $\Nbb_0$, such that $\tau_{k,\ell}$ is the $(a+1)$-th job in an active interval of $\tau_k$.
		We first show that $a\leq \tilde a_k$ by contraposition.
		If $a > \tilde a$, then we consider the job $\tau_{k,\ell-(a-\tilde a)}$.
		This is the $(\tilde a+1)$-th job in an active interval of $\tau_k$.
		Moreover, this job is one of the first $\xi$ jobs in $Seq$ and therefore has a response time of at most $T_k$ due to 2).
		Hence, $\tau_{k,\ell-(a-\tilde a)}$ is finished by time $r_{k,\ell-(a-\tilde a)+1}$.
		We conclude that $\tau_{k,\ell}$ is the $(a-\tilde a_k)$-th job in an active interval of $\tau_k$, which contradicts the minimality of $a$.
		We now choose $c:=b^a_k-a\cdot T_k + r_{k,\ell}$.
		Applying Lemma~\ref{lem:suff_sched_test_var_prep} with this $c$ shows that $\Rtilde^a_k(b^a_k)$ is a response time upper bound of $\tau_{k,\ell}$ as required.
	\end{proof}

We apply a similar search strategy as for the case with fixed analysis window. 
However, the values of $\Rtilde_k$ are computed through an additional loop over the values of $a$ until $\Rtilde^a_k \leq T_k$.
Algorithm~\ref{alg:sched_test_var} depicts an implementation of the schedulability test in pseudocode.
As the value of $\Rtilde^a_k$ can be between $T_k$ and $D_k$ for all iterations of $a$, the program may never return a result. 
To make the schedulability test deterministic, we introduce an additional parameter $max\_a$ which aborts the loop when even $a=max\_a$ gives no result.

\begin{algorithm}[t]
			\small
%			\footnotesize{}
			%\scriptsize
			\caption{Schedulability test with var.\@ analysis window.}
			\label{alg:sched_test_var}
			\begin{algorithmic}[1]
				\item[] \textbf{Input:} $\Tbb = \{\tau_1,\dots,\tau_n\}$, $(\Pi_1,\dots,\Pi_n)$, $\eta$, $max\_a$, $depth$
				\item[] \textbf{Output:} True: schedulable, False: no decision
				\item[]
				
				\State Order $\tau_1 ,\dots,\tau_n$, s.th. $D_1\geq \dots \geq D_n$.
				\State Set $\Rtilde_i := D_i$ for all $i$.
				
				\For{$i=1,2, \dots, depth$}
				\State $solved := True$
				\For{$k = 1, 2, \dots,n$}
				%			\State $\Phi = \floor{\frac{D_i}{\eta}}$
				\For{$a=0,1, \dots, max\_a$} \Comment{Different $a$.}
				\State $cand := [~]$; $step := \eta \cdot D_k$
				\Comment{Preparation.}
				\For{$b=0,step, 2\cdot step, \dots <aT_k + D_k$}
				\State $\triangleright$ Compute candidate:
%				\Comment{Compute candidate.}
				\State $cand.append(\Rtilde^a_k(b))$ from Equation~\eqref{eq:thm:suff_sched_test_var}
				\EndFor
				\State $\Rtilde^a_k := \min\limits (cand)$
				\Comment{Compare candidates.}
				\If{$\Rtilde^a_k\leq T_k$}\Comment{Check~cond.~1.}
				\State $\triangleright$ WCRT upper bound:
				\State $\tilde{a} := a$; $\Rtilde_k := \min_{a=0,\dots,\tilde a} \Rtilde^a_k$; 
				\textbf{break} 
%				\Comment{WCRT upper bound.}
				\EndIf
				\If{$\Rtilde^a_k>D_k$ or $a=max\_a$} \Comment{Check cond.~2.}
				\State $solved:=False$; $\Rtilde_k := D_k$; \textbf{break}
				\EndIf
				\EndFor
				\EndFor
				\EndFor
				\State \Return $solved$
			\end{algorithmic}
		\end{algorithm}

\subsection{Dominance of Fixed and Variable Analysis Window}
\label{sec:dominance}

At first glance, the analysis derived in Section~\ref{sec:test_arb} with the variable analysis window seems to improve the analysis from Section~\ref{sec:test_fixed} with fixed analysis window in all cases:
When setting $x = b_k + a\cdot T_k$ in Theorem~\ref{thm:suff_sched_test_var}, then the result is lower bounded by $\Rtilde_k(b_k)$ from Theorem~\ref{thm:suff_sched_test_fixed}.
However, both methods do not dominate each other, as demonstrated in the discussion of Figure~\ref{fig:eval_arbdl_dm} in Section~\ref{sec:evaluation}, due to the following reasons.
First, the analysis with variable analysis window can only analyze schedules where the length of active intervals is bounded.
More specifically, if the response-time upper bound $\Rtilde^a_k$ is in the interval $(T_k, D_k)$ for all $a$, then the schedulability test with the variable analysis window never deems the task schedulable.
Second, by setting $max\_a$ this effect is even intensified:
The analysis with variable analysis window has to find $\Rtilde^a_k \leq T_k$ even for some $a \leq max\_a$.
Third, the discretization using $\eta$ in Algorithm~\ref{alg:sched_test_fixed} and~\ref{alg:sched_test_var} ensures the same number of points for each analysis interval.
As a result, not all points $b+aT_k$ with $b$ from Algorithm~\ref{alg:sched_test_fixed} are checked during Algorithm~\ref{alg:sched_test_var} as well.

\section{Realization of EL Scheduling}
\label{sec:remark+realization}

To implement EDF-Like (EL) scheduling algorithms, we can exploit the existing EDF scheduling mechanisms to integrate the proposed relative priority points.
	Here we demonstrate the integration on two well-known real-time operating systems (RTOSes), i.e., RTEMS \cite{rtems} and {$\text{LITMUS}^{\text{RT}}$} \cite{litmusrt}, which officially support EDF scheduling.
	In general, the workflow of an EDF mechanism can be abstracted as follows:
	\begin{enumerate}
		\item At each job release, its priority should be identified with its absolute deadline, and necessary priority mapping operations should be performed. 
		\item The job context is placed into the ready queue according to its priority, which is commonly realized by a dedicated data structure, e.g., a red-black tree in RTEMS and a binomial heap in {$\text{LITMUS}^{\text{RT}}$}.
		\item One highest priority job should be executed. In case there are deadline ties, a regulation should be performed. This step also takes place if a job finishes its execution.
	\end{enumerate}
	Since EL scheduling only affects the job priority, the realization can be achieved by performing the priority mapping with the \emph{priority point} instead of the absolute deadline. This implies that  any underlying data structure for EDF scheduling, which sorts the jobs according to their absolute deadlines, can be directly adapted.
	Suppose the relative priority point for each task is given, we demonstrate the integration of these relative priority points for RTEMS (version 5.1) and {$\text{LITMUS}^{\text{RT}}$}:
	
	\paragraph{RTEMS} 
	In \texttt{Thread\_Configuration} we add a field \texttt{rel\_pp} to hold the given relative priority point.
	Two new fields \texttt{rel\_pp} and \texttt{release\_time} are added in \texttt{Thread\_Control} structure to keep the relevant modification minimal.
	To avoid corrupting the dependency of original priority mapping, we reuse the original mapping macro \texttt{SCHEDULER\_PRIORITY\_MAP()} and replace the original input \emph{deadline} by the priority point.
	
	\paragraph{{$\text{LITMUS}^{\text{RT}}$}} We add a new field \texttt{priority\_point} in \texttt{rt\_job} structure and a new field \texttt{rel\_pp} in \texttt{rt\_task} structure. 
	The calculation of \texttt{priority\_point} is performed at \texttt{setup\_release()}. 
	A new macro \texttt{higher\_priority\_point} is used to compare \texttt{priority\_point} over jobs.
	To ensure that the light-weight event tracing toolkit~\cite{Brandenburg07feather-trace:a} still functions correctly, we did not touch any relevant functions and macros for the absolute deadline.
	
	The functionality of both realizations are validated successfully on real platforms. The corresponding patches will be publicly available once the paper is accepted.
	Since the realization on both RTOSes only requires slight modifications,
	we conjecture that every RTOS supporting EDF scheduling can also realize EL scheduling. 
	Please note that the discussion is to demonstrate the applicability, and we do not recommend to replace TFP or FIFO with EL scheduling because of unnecessary operation overhead.

\section{Evaluation}
\label{sec:evaluation}

	\begin{figure}[t]
		\centering
		\begin{subfigure}{.5\linewidth}
			\centering
			\includegraphics[width=.95\linewidth]{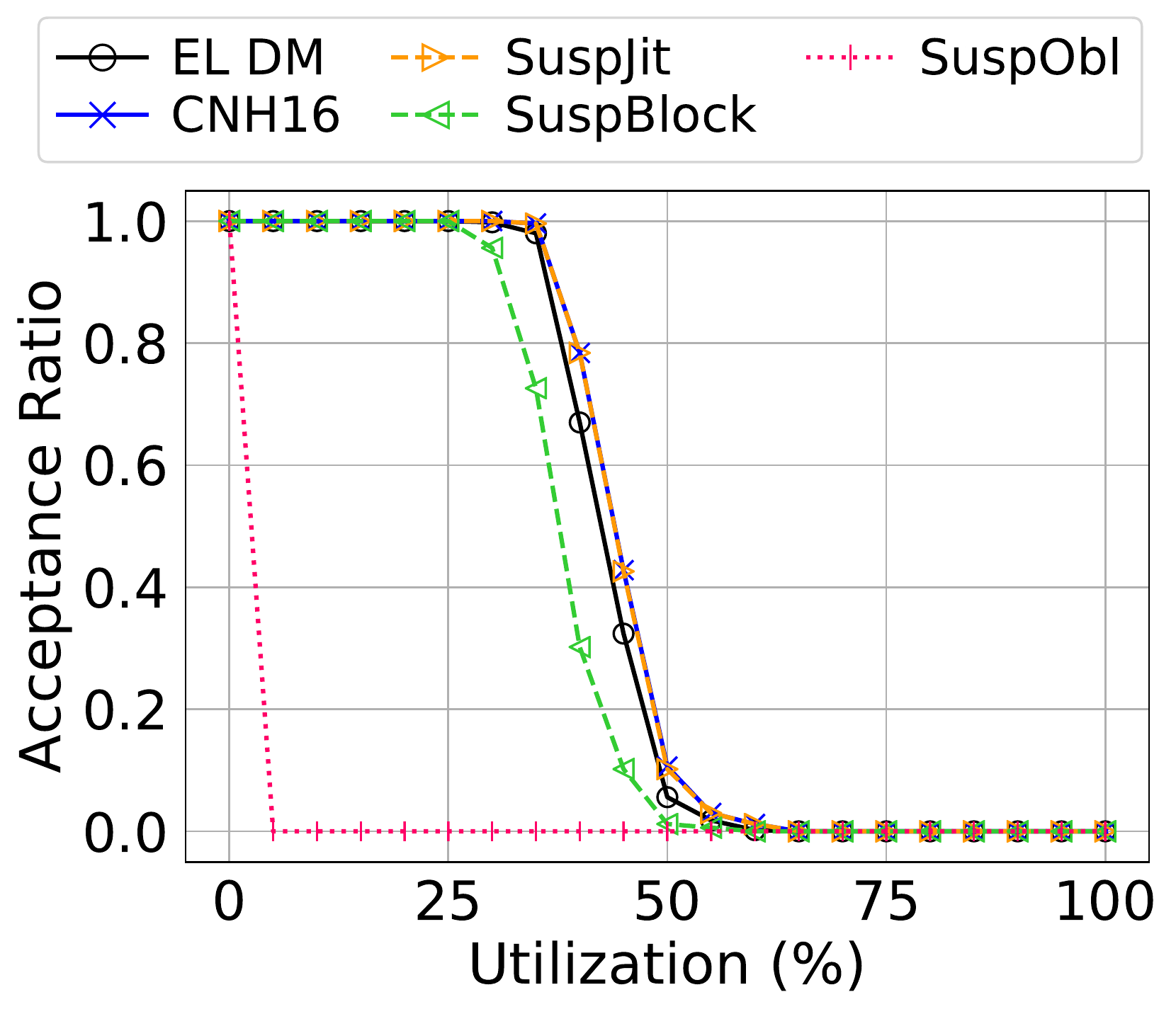}
			\caption{Deadline-Monotonic (DM).\label{fig:eval_dm}}
		\end{subfigure}%
		\begin{subfigure}{.5\linewidth}
			\centering
			\includegraphics[width=.95\linewidth]{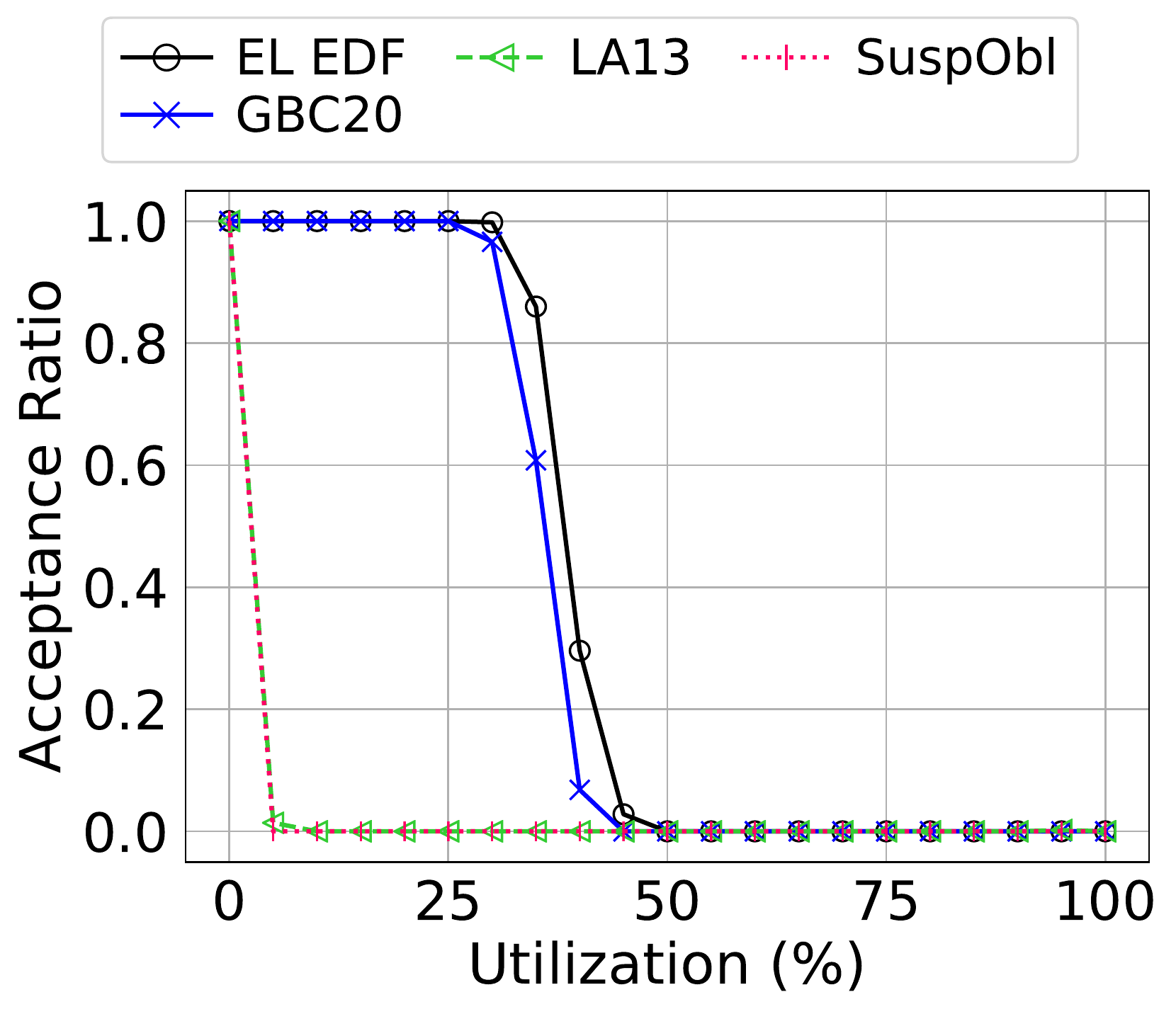}
			\caption{Earliest-Deadline-First (EDF).\label{fig:eval_edf}}
		\end{subfigure}
		\caption{Acceptance ratio of different schedulability tests. Our EDF-Like (EL) schedulability test (EL DM and EL EDF, black curve) performs similar to the state of the art. }
		\label{fig:dm_and_edf}
	\end{figure}
	
	\begin{figure}[t]
		\centering
		\begin{subfigure}{.5\linewidth}
			\centering
			\includegraphics[width=.95\linewidth]{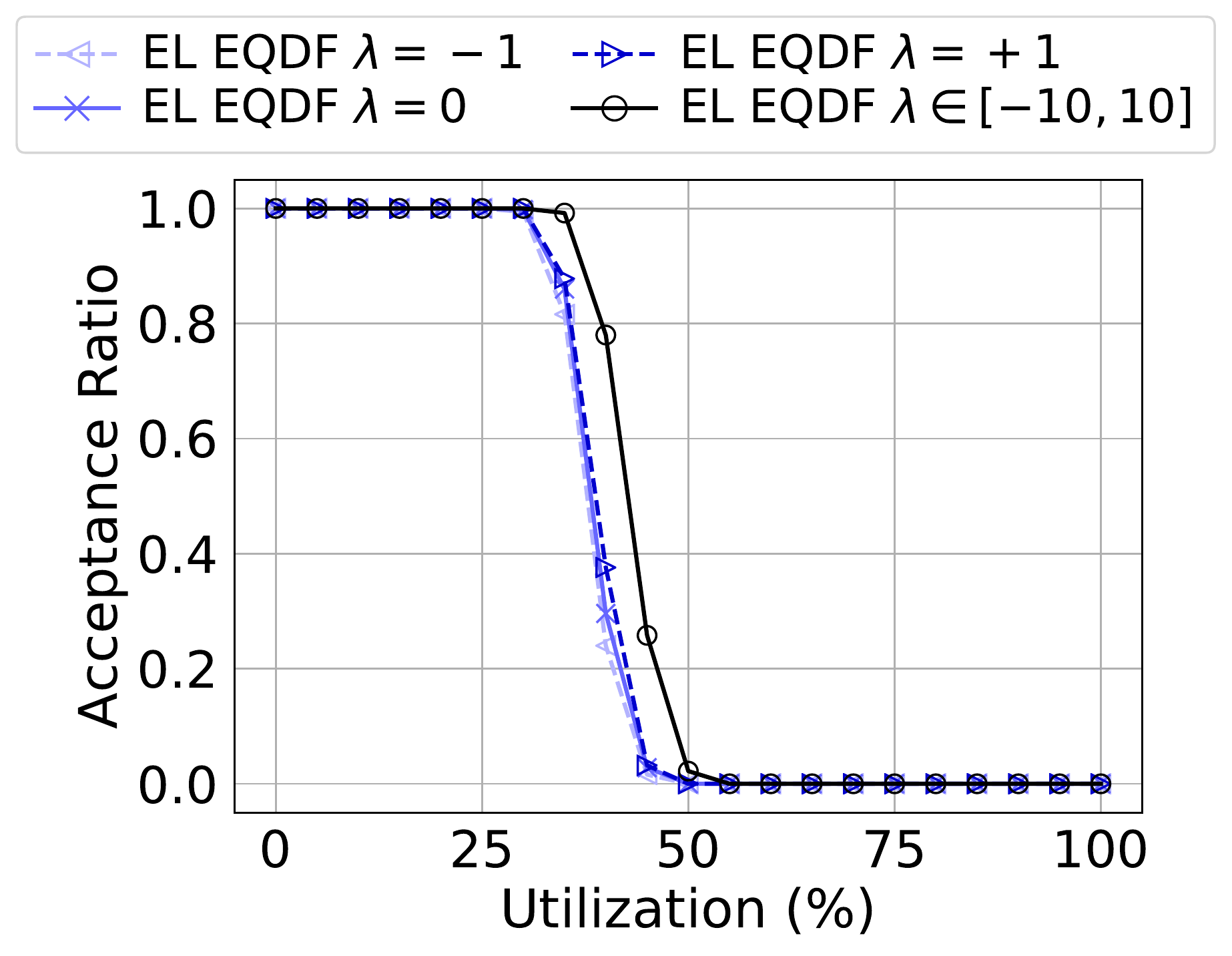}
			\caption{EQDF ($\Pi_i = D_i + \lambda C_i$).\label{fig:eval_eqdf}}
		\end{subfigure}%
		\begin{subfigure}{.5\linewidth}
			\centering
			\includegraphics[width=.95\linewidth]{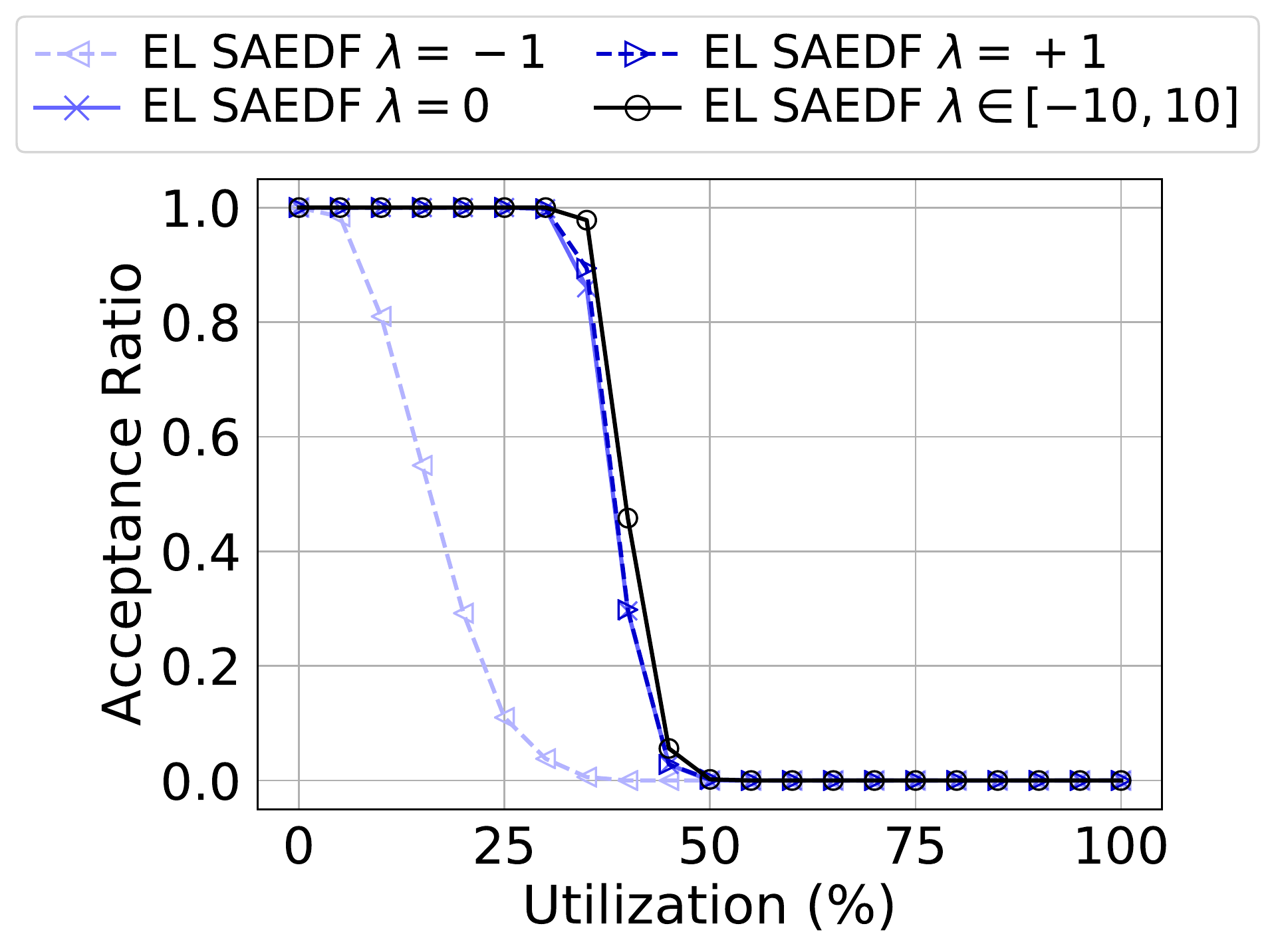}
			\caption{SAEDF ($\Pi_i = D_i + \lambda S_i$).\label{fig:eval_saedf}}
		\end{subfigure}
		\caption{Acceptance ratio of variants of EDF using our EDF-Like (EL) schedulability test. Choosing the best $\lambda \in [-10,10]$ for each task set (black line) improves standard EDF ($\lambda=0$).}
		\label{fig:eqdf_and_saedf}
	\end{figure}

	In this section, we evaluate the performance of our schedulability tests \textbf{(EL)} presented in Algorithm~\ref{alg:sched_test_fixed} for the fixed analysis window and in Algorithm~\ref{alg:sched_test_var} for the variable analysis window.
	More precisely we show that:
	\begin{enumerate}
		\item Our schedulability test performs similar to already existing schedulability tests for Deadline-Monotonic (DM) and Earliest-Deadline-First (EDF) scheduling.
		\item Our schedulability test can be used to compare different configurations of Earliest-Quasi-Deadline-First (EQDF) and suspension-aware EDF (SAEDF), presented in Section~\ref{sec:capabilities_and_limitations}.
		\item Our schedulability test exploits the optimism introduced when the deadline of tasks is extended over their minimum inter-arrival time.
	\end{enumerate}
	Please note that for 1) and 2) we do not distinguish between fixed and variable analysis window, since both schedulability tests coincide in the constrained deadline case, as explained in  Section~\ref{sec:test_arb}.
	To apply our schedulability test we use the configuration $\eta=0.01, depth=5$ and $max\_a=10$.
	In each figure, we present the \emph{acceptance ratio} which is the share of task sets that are deemed schedulable by the schedulability test under consideration.
	
	For the experiments, we synthesize $500$ task sets for each utilization from $0\%$ to $100\%$ in steps of $5\%$.
	For each task set we generate $50$ tasks.
	We first generate $50$ utilization values $U_i$ using the UUniFast~\cite{DBLP:journals/rts/BiniB05} method with the given total utilization goal, and then adopt the suggestion by Emberson~et~al.~\cite{emberson2010techniques} to pull the minimum inter-arrival time $T_i$ according to a log-uniform distribution from the interval $[1,100] [ms]$.
	The worst-case execution time is computed by $C_i = T_i \cdot U_i$ and the deadline is set to the minimum inter-arrival time $D_i = T_i$.
	For each task, we draw the maximum suspension time $S_i$ uniformly at random from $[0,0.5(T_i-C_i)]$.
	We assume that the tasks of each task set are ordered by their deadline.
	
	In Figure~\ref{fig:eval_dm} we apply EL with relative priority points $\Pi_i = \sum_{j=1}^i D_j$ to obtain a schedulability test for DM scheduling \textbf{(EL DM)}.
	We compare with the methods \emph{Suspension as Jitter} \textbf{(SuspJit)}~\cite[Page~163]{suspension-review-jj} and \emph{Suspension as Blocking} \textbf{(SuspBlock)}~\cite[Page~165]{suspension-review-jj}.
	Moreover, we compare with the Suspension-Oblivious Analysis \textbf{(SuspObl)}~\cite[Page~162]{suspension-review-jj} and the Unifying Analysis Framework from Chen, Nelissen and Huang \textbf{(CNH16)}~\cite{ChenECRTS2016-suspension} configured with three vectors according to Eq. (27), Lemma~15 and Lemma~16 of their paper.
	As depicted, our schedulability test performs similar to the state-of-the-art methods.
	
	In Figure~\ref{fig:eval_edf}, we compare our schedulability test \textbf{(EL EDF)} with state-of-the-art methods for EDF.
	We compare with the method by Liu and Anderson \textbf{(LA13)}~\cite{DBLP:conf/ecrts/LiuA13}.
	Moreover, we present the schedulability test by Günzel, von der Brüggen and Chen \textbf{(GBC20)}~\cite{guenzel2020sched_test_edf} and the Suspension-Oblivious Analysis \textbf{(SuspObl)}~\cite[Section~III.A]{guenzel2020sched_test_edf}.
	The method from Dong and Liu~\cite{DBLP:conf/rtss/DongL16} is not presented as it is dominated by \textbf{SuspObl}, as shown in~\cite{guenzel2020sched_test_edf}.
	\textbf{EL EDF} improves the state of the art.

    	\begin{figure}[t]
		\centering
		\begin{subfigure}{.5\linewidth}
			\centering
			\includegraphics[width=.95\linewidth]{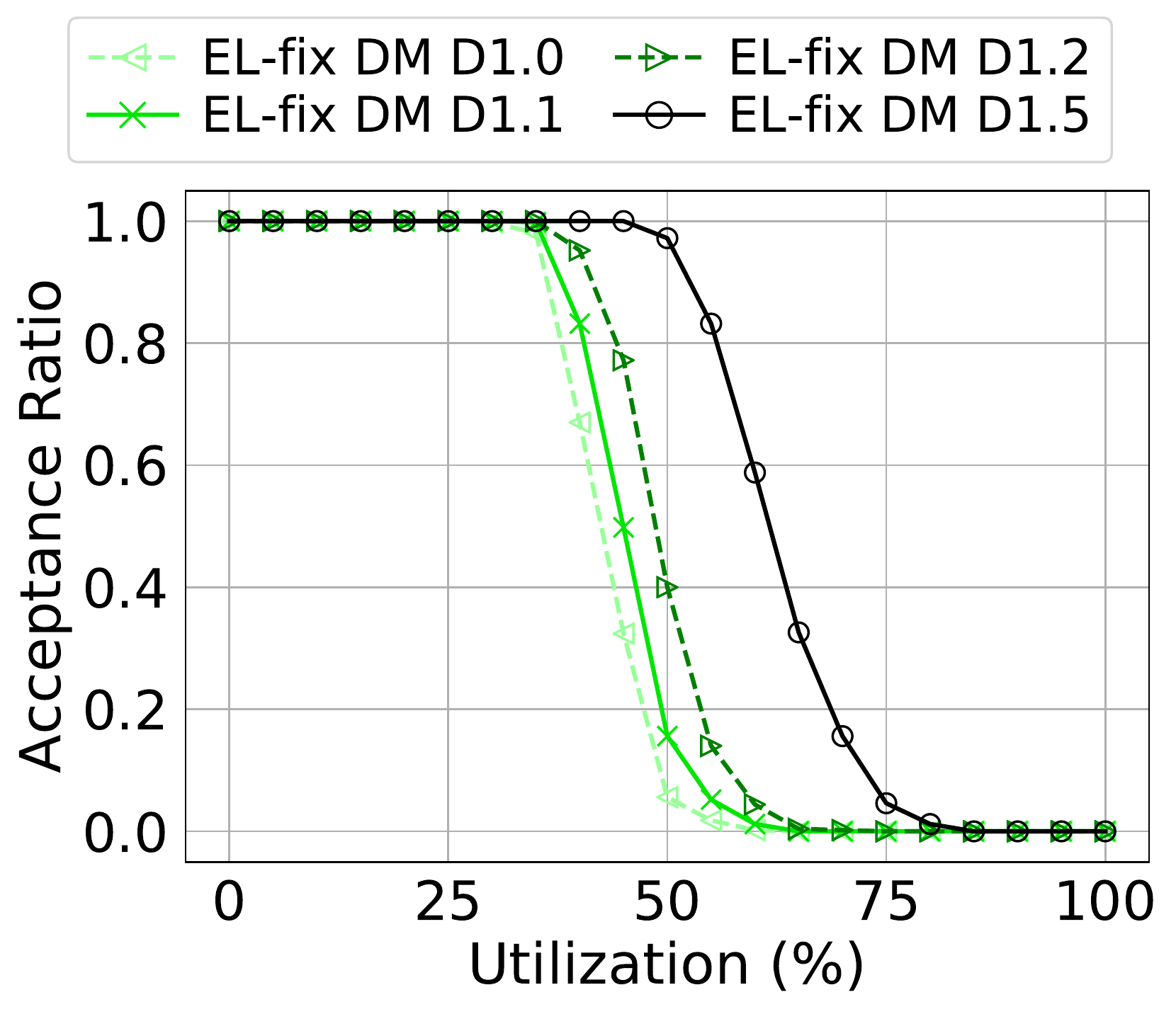}
			\caption{Fixed analysis window.\label{fig:eval_arbdl_dm_fix}}
		\end{subfigure}%
		\begin{subfigure}{.5\linewidth}
			\centering
			\includegraphics[width=.95\linewidth]{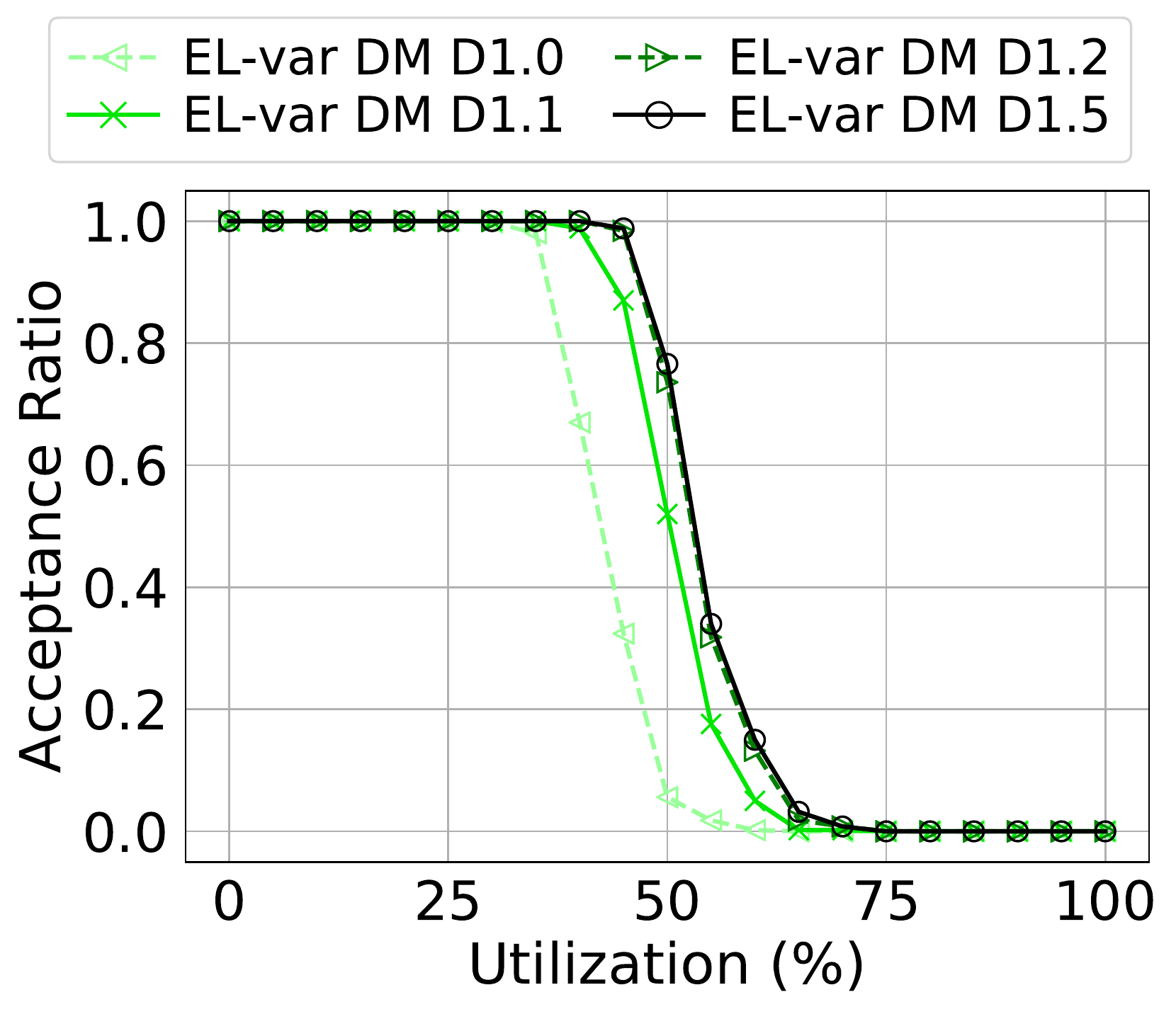}
			\caption{Variable analysis window.\label{fig:eval_arbdl_dm_var}}
		\end{subfigure}
		\caption{Arbitrary deadline evaluation for Deadline-Monotonic~(DM) scheduling.}
		\label{fig:eval_arbdl_dm}
	\end{figure}
	
	In Figure~\ref{fig:eqdf_and_saedf} the performance of our schedulability test is presented for different configurations for Earliest-Quasi-Deadline-First (EQDF) ($\Pi_i = D_i + \lambda C_i)$ and for suspension-aware EDF (SAEDF) ($\Pi_i = D_i + \lambda S_i$).
	As depicted, choosing $\lambda$ to be the best integer in $[-10,10]$, improves acceptance ratio compared to the standard EDF with $\lambda= 0$, especially for EQDF.
	
	In Figures~\ref{fig:eval_arbdl_dm} and~\ref{fig:eval_arbdl_edf}, 
	the performance of our schedulability test for arbitrary deadlines is presented.
	More specifically, we set the deadline to $x=1.0,1.1,1.2,1.5$ times the minimum inter-arrival time \textbf{(Dx)} and apply our schedulability test.
	We see that both the fixed and the variable analysis window lead to better acceptance ratios in certain scenarios, depending on the size of $x$ and the scheduling algorithm under analysis.
	The non-dominance discussion from Section~\ref{sec:dominance} can be observed in Figure~\ref{fig:eval_arbdl_dm} for D1.2 and D1.5.

	\begin{figure}[t]
		\centering
		\begin{subfigure}{.5\linewidth}
			\centering
			\includegraphics[width=.95\linewidth]{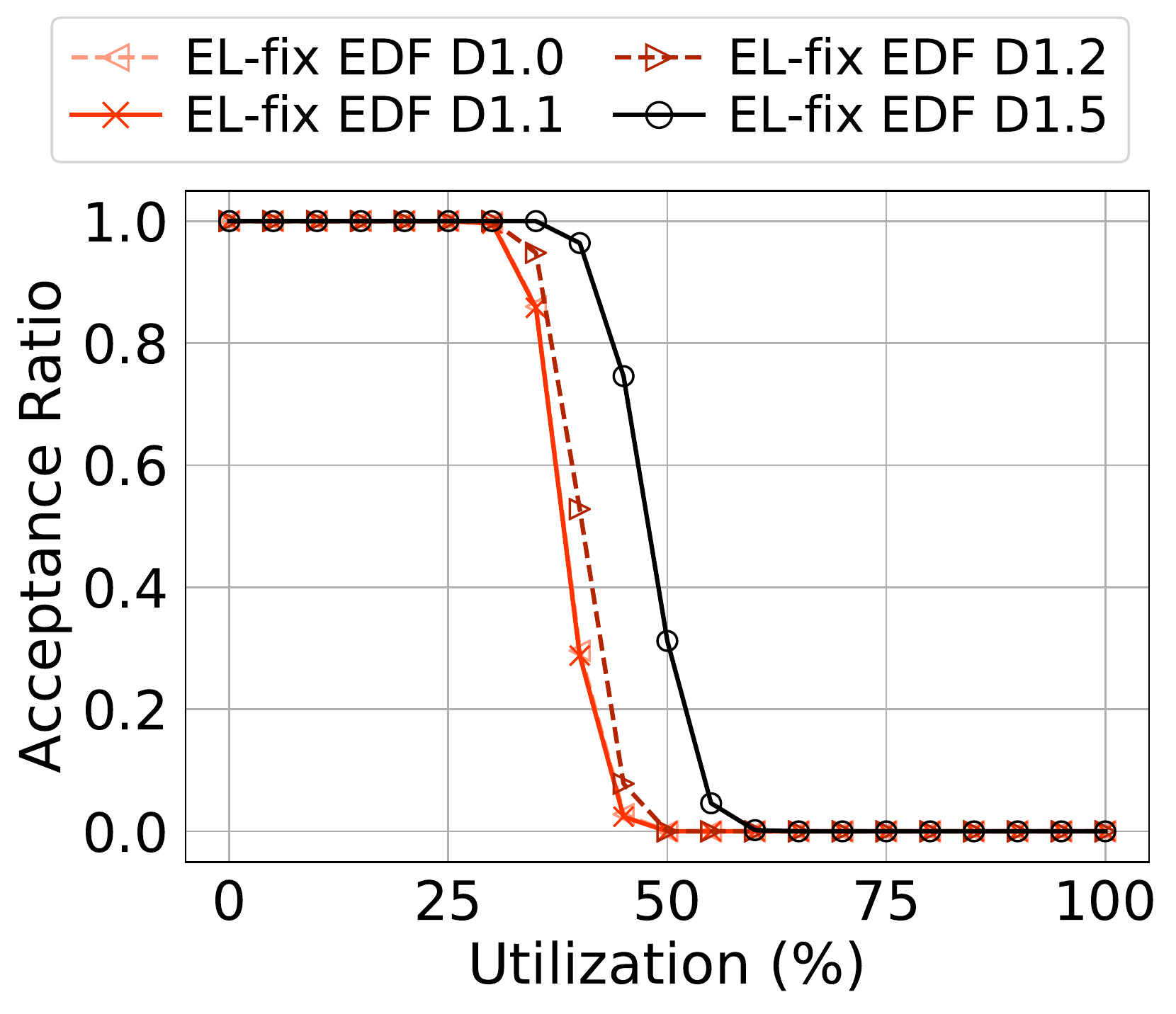}
			\caption{Fixed analysis window.\label{fig:eval_arbdl_edf_fix}}
		\end{subfigure}%
		\begin{subfigure}{.5\linewidth}
			\centering
			\includegraphics[width=.95\linewidth]{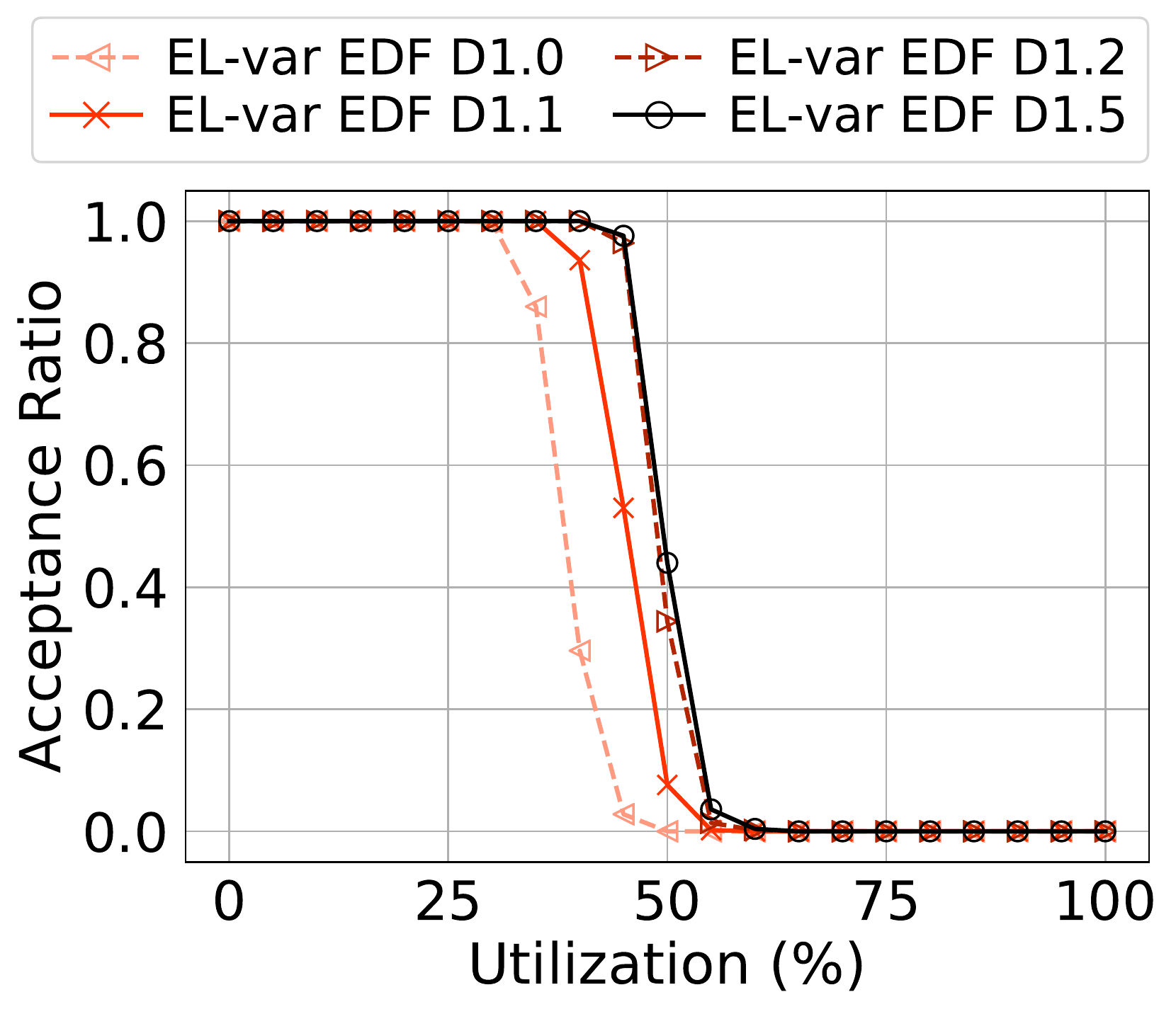}
			\caption{Variable analysis window.\label{fig:eval_arbdl_edf_var}}
		\end{subfigure}
		\caption{Arbitrary deadline evaluation for Earliest-Deadline-First (EDF) scheduling.}
		\label{fig:eval_arbdl_edf}
	\end{figure}

\begin{figure}[t]
	\centering
	\begin{subfigure}{.5\linewidth}
		\centering
		\includegraphics[width=.95\linewidth]{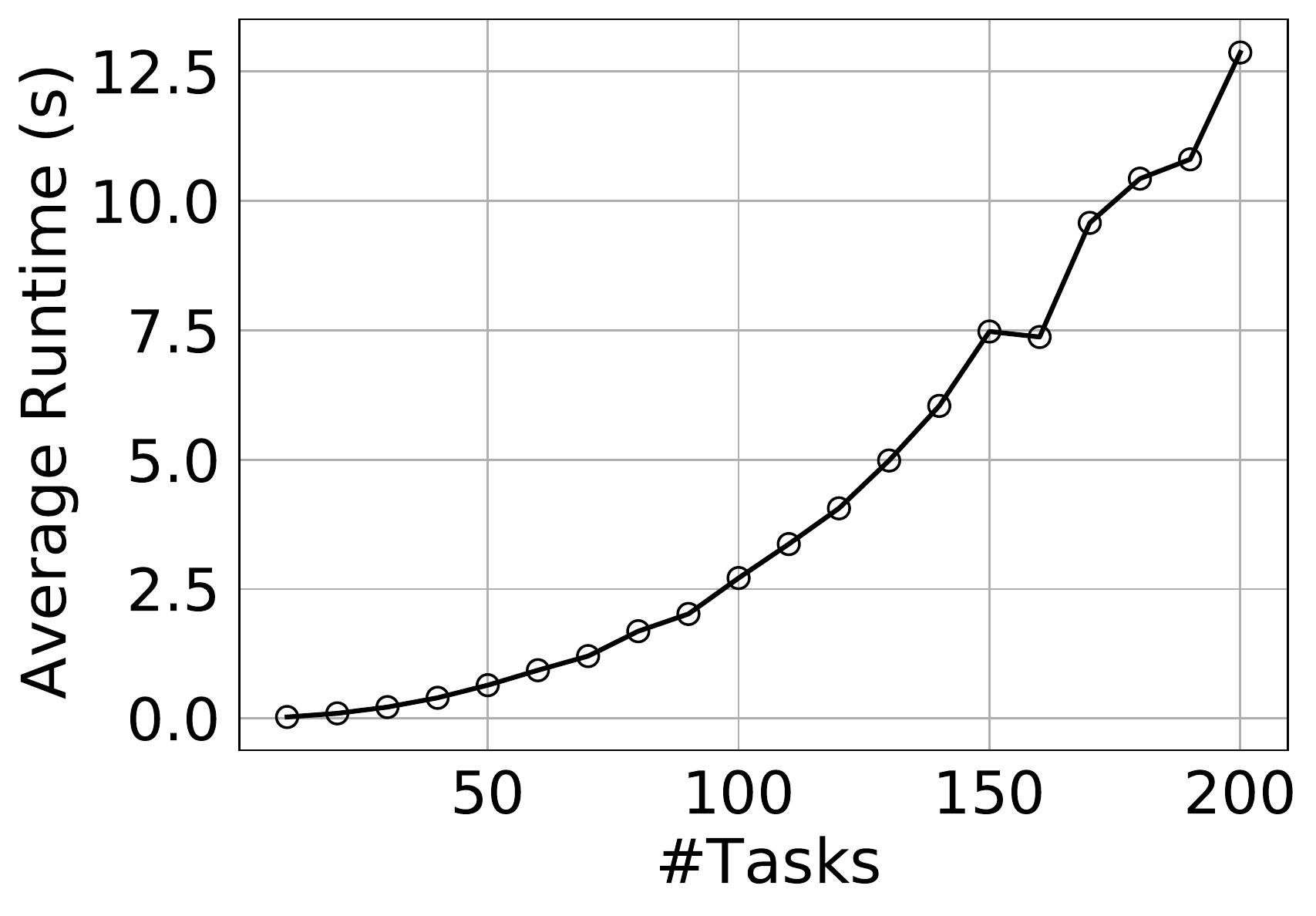}
		\caption{Average runtime.}
	\end{subfigure}%
	\begin{subfigure}{.5\linewidth}
		\centering
		\includegraphics[width=.95\linewidth]{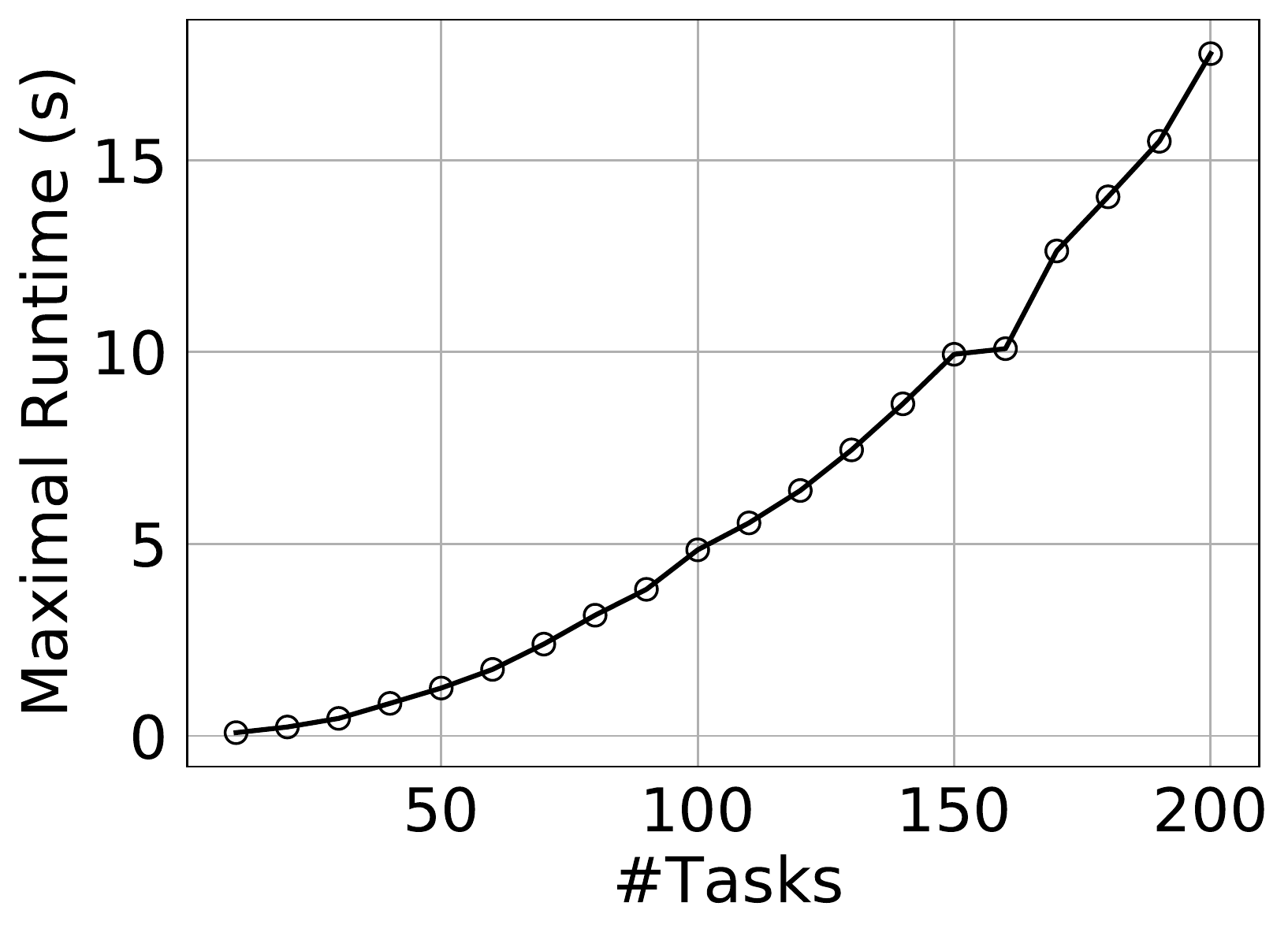}
		\caption{Maximal runtime.}
	\end{subfigure}
	\caption{Runtime (seconds) per task set of our schedulability test (EL EDF).}
	\label{fig:runtime_results}
\end{figure}

Furthermore, we study the impact of the number of tasks on the \textbf{runtime} of our analysis.
In this regard, we create 100 task sets for each utilization in $0\%$ to $100\%$ in steps of $10\%$ and measure the runtime that it takes to receive a schedulability decision.
As an example we show the results for EDF scheduling $(\Pi_i = D_i)$.
For other relative priority points, the runtime is comparable.
To obtain the measurements, we run an implementation with Python3 on a machine with 2x AMD EPYC 7742 running Linux, i.e., in total we have 256 threads with 2,25GHz and 256GB RAM.
Each of the measurements runs on one independent thread.
The results are presented in Figure~\ref{fig:runtime_results}.
We observe that the runtime of our method grows fast. 
However, even with $200$ tasks per task set, our schedulability test takes on average $12.87$ seconds and at most $17.77$ seconds to return the result for one task set.

\section{Conclusion}
\label{sec:conclusion}

In this work, we study EDF-Like (EL) scheduling algorithms.
Through an examination of different analysis intervals we provide two versions of a suspension-aware schedulability test, valid for all EL scheduling algorithms, even for arbitrary-deadline tasks.
We provide the first suspension-aware schedulability test to handle EL scheduling for arbitrary-deadline tasks.
In particular, this is also the first suspension-aware schedulability test for arbitrary-deadline tasks under First-In-First-Out (FIFO) scheduling, Earliest-Quasi-Deadline-First (EQDF) scheduling and Suspension-Aware EDF~(SAEDF) scheduling.

% bib
%\pagebreak
%\clearpage
\bibliographystyle{abbrv}
\bibliography{real-time}

\end{document}